\documentclass[12pt]{amsart}

\usepackage{amsfonts}
\usepackage{amsmath}
\newtheorem{thm}{Theorem}[section]

\newtheorem{cor}[thm]{Corollary}
\newtheorem{lem}[thm]{Lemma}
\newtheorem{prop}[thm]{Proposition}
\newtheorem*{prob*}{Problem}
\newtheorem*{thm*}{Theorem}

\theoremstyle{definition}
\newtheorem{defn}[thm]{Definition}

\newtheorem*{defn*}{Definition}
\newtheorem{rem}[thm]{Remark}

\newtheorem*{rem*}{Remark}
\numberwithin{equation}{section}

\newcommand{\C}{\mathbb C}
\newcommand{\R}{\mathbb R}

\DeclareMathOperator{\const}{const}

\DeclareMathOperator{\Mat}{Mat}

\newcommand{\Tr}{\mathop{\mathrm{Tr}}}

\newcommand{\hP}{\widehat{P}}

\newcommand{\hI}{\widehat{I}}
\newcommand{\hK}{\widehat{K}}

\begin{document}
\title[The product of two dependent random matrices]
 {\bf{Dropping the independence: singular values for products of two coupled random matrices}}

\author{Gernot Akemann}
\address{Fakult\"at f\"ur Physik, Universit\"at Bielefeld,  Postfach 100131, D-33501 Bielefeld, Germany, and Universit\'e Paris-Sud, CNRS, LPTMS, UMR 8626, B{\^a}t. 100, Orsay F-91405, France} \email{akemann@physik.uni-bielefeld.de}

\author{Eugene Strahov}
\address{Department of Mathematics, The Hebrew University of
Jerusalem, Givat Ram, Jerusalem
91904, Israel}\email{strahov@math.huji.ac.il}
\thanks{The first author (G.~A.) is supported partly 
by Investissements d'Avenir du LabEx PALM
(ANR-10-LABX-0039-PALM) and 
by the SFB$|$TR12 ``Symmetries
and Universality in Mesoscopic Systems'' of the German research council DFG.
The second author (E.~S.) is supported in part by the Hebrew University  Grant  ``Non-hermitian random matrices''  No.\ 0337592.\\ }

\commby{}
\begin{abstract}
We study the singular values of the product of two coupled rectangular random matrices as a determinantal point process. 
Each of the two factors is given by a parameter dependent linear combination of two independent, complex Gaussian random matrices, which is equivalent to a coupling of the two factors via an Itzykson-Zuber term. 
We prove that the squared singular values of such a product form a biorthogonal ensemble and establish its exact solvability.
The parameter dependence allows us to interpolate between the singular value statistics of the Laguerre ensemble and that of the product of two independent complex Ginibre ensembles which are both known. 
We give exact formulae for the correlation kernel in terms of a complex double contour integral, suitable for the subsequent asymptotic analysis. In particular, we derive a Christoffel-Darboux type formula for the correlation kernel, based on a five term recurrence relation for our biorthogonal functions. It enables us to find its scaling limit at the origin representing a hard edge. The resulting limiting kernel coincides with the
universal Meijer G-kernel found by several authors in different ensembles.
We show that the central limit theorem holds for the linear statistics of the singular values and give the
limiting variance explicitly.
\end{abstract}
\maketitle
\section{Introduction}
A remarkable feature of products of independent complex Gaussian matrices, i.e. independent matrices with i.i.d. standard complex Gaussian entries, is the exact solvability of the statistical properties of their eigenvalues
and singular values. Indeed, it was shown  in  
\cite{Akemann1} that the eigenvalues of such products form  a determinantal
point processes in $\C$. The behaviour of singular values for such products was studied in  
\cite{AkemannKieburgWei,AkemannIpsenKieburg}, where it was observed that its (squared) singular values 
also form a determinantal point process in $\R_{\geq 0}$. The correlation kernels of these two different determinantal point processes can be written
explicitly in terms of  Meijer G-functions, with suitable choices of parameters.

These results 
have opened the possibility to investigate products of independent complex Gaussian matrices on the same level as the well-known classical ensembles of Random Matrix Theory, such as the Ginibre ensemble and the Laguerre ensemble.
We refer the reader to the books by
Anderson,  Guionnet and Zeitouni \cite{Anderson}, 
and by Forrester \cite{ForresterLogGases} 
for an introduction to Random Matrix Theory,
as well as to \cite{ABD} for a compilation of its most recent applications.

The study of products of random matrices goes back to Furstenberg and Kesten \cite{FK} who were interested in its Lyapunov exponents that characterise dynamical systems. Many statistical mechanics applications have been summarised in the book by Crisanti, Paladin, and Vulpiani \cite{CPV}, and most recent examples for applications include telecommunications \cite{Ralf} and combinatorics \cite{Karol}. A very particular case of the product of two coupled matrices was applied to Quantum Chromodynamics (QCD) with chemical potential in \cite{Osborn}, where the complex eigenvalue spectrum was determined. This example will be important for our paper, due to the coupling of the matrices.  

The recent rapid development on products of matrices  is summarised in the review 
\cite{AkemannIpsen}, to where we refer for details and references. 
In particular, in  the work by Kuijlaars and Zhang \cite{KuijlaarsZhang} a new class of so-called Meijer G-kernels was found near the origin, representing a hard edge. The name alludes to the appearance of the Meijer G-function. This kernel generalises the Bessel kernel and contains the kernels of Borodin \cite{BorodinBO}, as pointed out by Kuijlaars and Stivigny \cite{KS}. It is universal 
as it remains unchanged when multiplying by an additional independent inverse complex Gaussian matrices as shown by Forrester \cite{ForresterProductWishart} or by an additional truncated unitary matrix as shown by Kuijlaars and Stivigny \cite{KS}. Furthermore, it appears in the Cauchy two-matrix model \cite{Bertola} and its multi-matrix extension \cite{MB2} of Bertola and coworkers.
Because the Cauchy two-matrix model was used recently to solve the (Laplace transform) of a matrix model with Bures measure by Forrester and Kieburg \cite{PeterMario}, this kernel enjoys  applications to quantum density matrices. And we will also find this limiting Meijer G-kernel for two independent matrices, starting from two coupled random matrices.
It was shown in Kuijlaars and Zhang \cite{KuijlaarsZhang} that the class of kernels  
is integrable in the sense of Its, Isergin, Korepin, and Slavnov  \cite{Its}.
This enabled the description of the squared singular values by Hamiltonian equations \cite{Strahov}.
For a survey on integrable operators see Deift \cite{Deift1}.
Furthermore, contact was made to questions from Gaussian analytic functions in \cite{AkemannIpsen,AkemannIpsenStrahov} by studying the asymptotics of gap and overcrowding probabilities. For very recent results on determinantal point processes related to products
of independent complex Gaussian matrices we refer the reader to Kuijlaars  \cite{Kuijlaars},  Forrester and Wang  \cite{ForresterWang}, and Forrester and  
Liu \cite{ForresterLiu}.

Two questions arise naturally: What happens when the assumption of a Gaussian distribution of matrix elements is dropped?
When the matrices in the product are independent, but not necessarily Gaussian, a number of results for the statistics of eigenvalues and of singular
values in the\textit{ global} asymptotic regime is available. 
The paper by O`Rourke and Soshnikov \cite{Rourke} gives an analogue of the circular law for the product
of a finite number of non-Hermitian random matrices, generalising the result by Burda, Janik, and Waclaw \cite{Burda3}. For a description
of the statistics of singular values of products of independent matrices, and, in particular, for the Central Limit Theorem for the squared singular values
we refer the reader to the papers by 
G$\ddot{\mbox{o}}$tze, Tikhomirov and their co-workers \cite{AlexeevGotzeTikhomirov,GotzeNaumovTikhomirov,GotzeKostersTikhomirov}. 
Results on the \textit{local} statistics for products of independent matrices with non-Gaussian entries are still not available, to the best of our knowledge. This is no doubt due to the lack of integrability in the non-Gaussian case.

The second question is whether some of the above results 
can be extended to those of coupled random matrices. 
Here we consider a product of two \textit{dependent} matrices, and concentrate on the statistics of the squared singular values.
Such random matrices appeared first in the work by  Osborn \cite{Osborn} in the context of QCD with a baryon chemical potential $\mu$ 
as follows:
\begin{equation}
\label{OsbornMatrix}
D=\left(\begin{array}{cc}
          0 & iA+\mu B \\
          iA^*+\mu B^* & 0
        \end{array}
\right).
\end{equation}
Here $A$ and $B$ are rectangular independent matrices with i.i.d standard complex Gaussian entries, and $\mu\in[0,1]$ is a dimensionless parameter.
The motivation to consider (\ref{OsbornMatrix}) comes from the observation that the QCD Dirac operator $D$ 
has this off-diagonal block form in the so-called chiral basis.  For the random matrix application to QCD we refer to the review by Verbaarschot and Wettig \cite{JT}, see also chapter 32 in \cite{ABD} by Verbaarschot.
In \cite{Osborn} the correlations of complex eigenvalues of $D$ were determined, which is equivalent to determining the eigenvalues of the product matrix 
$Y=X_1X_2$, with 
$X_1=\left(iA+\mu B\right)$ and $X_2=\left(iA^*+\mu B^*\right)$. The change of variables from matrices $A, B$ to $X_1, X_2$ reveals that the latter are coupled by an Itzykson-Zuber term, 
in addition to their Gaussian weight.
Very recently it has been suggested in \cite{TT1,TT2} to study the singular values of the Dirac operator in QCD and QCD-like theories
instead, in order to better understand the high-density regime.
This is one of the motivations for us to study the (squared) singular values of the product matrix $Y$. 
Apart from this physical interpretation the parameter $\mu$ allows to interpolate between the classical Laguerre ensemble at $\mu=0$ solved by orthogonal Laguerre polynomials and the recent solution of the product of two independent Gaussian random matrices at $\mu=1$ given in terms of biorthogonal functions.

This paper is organised as follows. In Section \ref{SectionMuGaussianComplexMatrices} we define the notion of $\mu$-dependent Gaussian complex random matrices, making this notion of interpolation and of its limits more precise. 
We state our main results in Section \ref{SectionStatementOfResults}. In particular 
we demonstrate the exact solvability of the statistical properties of the singular values of the product matrix $Y$ for arbitrary parameter values $\mu$: the joint probability density function of the squared singular values is a determinantal process on $\R_{\geq 0}$ and can be computed explicitly in terms of modified Bessel functions of first and second kind.
This determinantal point process is a biorthogonal ensemble in the sense of Borodin \cite{BorodinBO}.
For this parameter dependent ensemble we derive different
formulae for the correlation kernel including a Christoffel-Darboux type formula and a double complex contour integral representation. We compute the hard edge scaling limit at the origin, and we obtain a Central Limit Theorem for fluctuations of linear statistics.
Sections \ref{ProofsInitial}-\ref{proofsFinite} and Appendix \ref{limitA} contain the proofs of our statements.
\ \\
\textbf{Acknowledgements.}
We are grateful to Percy Deift for discussions, and to Jonathan Breuer for a clear explanation of the results in Ref. \cite{BreuerDuits} to us. One of us (G.A.) would like to thank the LPTMS Orsay for hospitality where part of these results were finalised.

\section{
Parameter dependent Gaussian complex matrices}\label{SectionMuGaussianComplexMatrices}
Before we present our results we will define a family of parameter dependent coupled Gaussian random variables, and the corresponding notion for random matrices.
By this we mean the following.
\begin{defn}
Let $\mu\in (0,1)$, $\alpha(\mu)=\frac{1+\mu}{2\mu}$, and $\delta(\mu)=\frac{1-\mu}{2\mu}$. We will  refer to two complex random variables, $z$ and $\xi$ as to  \textit{$\mu$-dependent Gaussian complex variables} if the joint density of these variables is given by
$$
\rho(z,\xi)=\frac{1}{\pi^2\mu}\exp\left[-\alpha(\mu)(z\bar{z}+\xi\bar{\xi})+\delta(\mu)(z\xi+\bar{z}\bar{\xi})\right].
$$
\end{defn}
\begin{defn}
Let
$$
X_1=\left(\begin{array}{ccc}
            X_{1,1}^{(1)} & \ldots & X_{1,M}^{(1)} \\
            \vdots &  &  \\
            X_{N,1}^{(1)} & \ldots & X_{N,M}^{(1)}
          \end{array}
\right),\;\;\; X_2=\left(\begin{array}{ccc}
            X_{1,1}^{(2)} & \ldots & X_{1,N}^{(2)} \\
            \vdots &  &  \\
            X_{M,1}^{(2)} & \ldots & X_{M,N}^{(2)}
          \end{array}
\right)
$$
be two matrices whose complex random entries are defined by the following conditions
\begin{itemize}
  \item $X_{i,j}^{(1)}$, $1\leq i\leq N$, $1\leq j\leq M$ are independent;
  \item $X_{i,j}^{(2)}$, $1\leq i\leq M$, $1\leq j\leq N$ are independent;
  \item For each $1\leq i\leq N$, and for each $1\leq j\leq M$ the pair $(X_{i,j}^{(1)}, X_{j,i}^{(2)})$ is a pair of $\mu$-dependent
  Gaussian complex random variables.
\end{itemize}
We will refer to  such random matrices  $X_1$ and $X_2$ as to \textit{$\mu$-dependent Gaussian complex random matrices}.
\end{defn}
Alternatively,  we can define $\mu$-dependent Gaussian complex random matrices as follows. Let $\Mat(\C, N\times  M)$  denote the space of $N\times M$ complex matrices $X_1$,
and $\Mat(\C, M\times  N)$ denote the space of $M\times N$ complex random matrices $X_2$. We consider the probability distribution $P_{N,M}(X_1,X_2)dX_1dX_2$ on the Cartesian
product of $\Mat(\C, N\times  M)$ and $\Mat(\C, M\times  N)$
\begin{equation}\label{JointDensityX1X2}
\begin{split}
P_{N,M}(X_1,X_2)dX_1dX_2=&c\cdot 
\exp\left[-\alpha(\mu)\Tr\left(X_1X_1^*+X_2^*X_2\right)
+\delta(\mu)\Tr\left(X_1X_2+X_2^*X_1^*\right)\right]\\
&\times\prod\limits_{i=1}^N\prod\limits_{j=1}^Md{X_{i,j}^{(1)}}^Rd{X_{i,j}^{(1)}}^I\prod\limits_{i=1}^M\prod\limits_{j=1}^Nd{X_{i,j}^{(2)}}^Rd{X_{i,j}^{(2)}}^I,
\end{split}
\end{equation}
where $X_{i,j}^{(1)}={X_{i,j}^{(1)}}^R+i{X_{i,j}^{(1)}}^I$, $X_{i,j}^{(1)}={X_{i,j}^{(1)}}^R+i{X_{i,j}^{(1)}}^I$ denote the sums of the real and imaginary parts of the matrix entries $X_{i,j}^{(1)}$ and $X_{i,j}^{(2)}$, and $c$ is a normalising constant. 
The second term in the exponent proportional to $\delta(\mu)$ is nothing else than the Itzykson-Zuber term (for non-hermitian matrices) coupling the two matrices\footnote{However, because we will be interested in the singular values of the product matrix $X_1X_2$, we will not use their integration formula \cite{HC,IZ} for this term.}.
We have
$$
\Tr(X_1X_1^*)=\sum\limits_{i=1}^{N}\sum\limits_{j=1}^MX_{i,j}^{(1)}\overline{X_{i,j}^{(1)}},\;\;
\Tr(X_2^*X_2)=\sum\limits_{i=1}^{N}\sum\limits_{j=1}^MX_{j,i}^{(2)}\overline{X_{j,i}^{(2)}},
$$
and
$$
\Tr(X_1X_2)=\sum\limits_{i=1}^{N}\sum\limits_{j=1}^MX_{i,j}^{(1)}X_{j,i}^{(2)},\;\;\Tr(X_2^*X_1^*)=\sum\limits_{i=1}^{N}\sum\limits_{j=1}^M\overline{X_{i,j}^{(1)}}
\; \overline{X_{j,i}^{(2)}}.
$$
Therefore the formula for $P_{N,M}(X_1,X_2)dX_1dX_2$ can be rewritten as
\begin{equation}
\begin{split}
P_{N,M}(X_1,X_2)dX_1dX_2=&c\cdot\prod\limits_{i=1}^N\prod\limits_{j=1}^M
 e^{-\alpha(\mu)\left(X_{i,j}^{(1)}\overline{X_{i,j}^{(1)}}+X_{j,i}^{(2)}\overline{X_{j,i}^{(2)}}\right)
 +\delta(\mu)\left(X_{i,j}^{(1)}X_{j,i}^{(2)}+\overline{X_{i,j}^{(1)}}
\; \overline{X_{j,i}^{(2)}}\right)}\\
&\times d{X_{i,j}^{(1)}}^Rd{X_{i,j}^{(1)}}^Id{X_{j,i}^{(2)}}^Rd{X_{j,i}^{(2)}}^I.
\end{split}
\end{equation}
It is clear from the formula just written above that $P_{N,M}(X_1,X_2)$ is indeed the probability distribution of the $\mu$-dependent Gaussian complex matrices $X_1$
and $X_2$. In addition, note that the normalising constant $c$ is equal to
$$
c=\frac{1}{\left(\pi^2\mu\right)^{NM}}.
$$
\begin{prop}
Let $A$, $B$ be two independent matrices of size $N\times M$ with i.i.d
standard complex Gaussian entries. Define the  random matrices $X_1$ and $X_2$ as
\begin{equation}\label{GaussSums}
X_1=\frac{1}{\sqrt{2}}\left(A-i\sqrt{\mu}B\right),\;\; X_2=\frac{1}{\sqrt{2}}\left(A^*-i\sqrt{\mu}B^*\right).
\end{equation}
Then the  matrices $X_1$ and $X_2$ 
are $\mu$-dependent Gaussian complex random matrices.
\end{prop}
\begin{proof}This can be checked by direct calculation.
\end{proof}

\section{Statement of results}\label{SectionStatementOfResults}
\subsection{The joint probability density function}
Our first result is an explicit formula for the joint probability density function for the squared singular values of the random matrix $X_1X_2$.
Recall that the modified Bessel function of the first kind $I_{\kappa}(z)$ is defined by
\begin{equation}\label{BesselFunctionI}
I_{\kappa}(z)=
\sum\limits_{m=0}^{\infty}\frac{1}{m!\Gamma(\kappa+m+1)}
\left( \frac{z}{2}\right)^{2m+\kappa},
\end{equation}
and the modified Bessel function of the second kind $K_{\kappa}(z)$ can be defined by the integral formula
\begin{equation}
K_{\kappa}(z)=\frac{\Gamma\left(\kappa+\frac{1}{2}\right)(2z)^{\kappa}}{\sqrt{\pi}}\int\limits_0^{\infty}\frac{\cos (t)dt}{(t^2+z^2)^{\kappa^2+\frac{1}{2}}},
\end{equation}
see, for example, Gradshteyn and Ryzhik \cite{GradshteinRyzhik}.
\begin{thm} \label{PropositionTheMainJointProbabilityDensity}  Let $X_1\in\Mat\left(\C, N\times M\right)$ and $X_2\in\Mat\left(\C, M\times N\right)$ be two $\mu$-dependent Gaussian
complex matrices. Assume that $M\geq N$, and set
$$
\nu=M-N.
$$
Then the joint probability density function for
the squared singular values $y_1$, $\ldots$, $y_N$ of the matrix $Y=X_1X_2$ is given by
\begin{equation}\label{TheMainJointProbabilityDensityFunction}
\begin{split}
&P(y_1,\ldots,y_N)=\frac{1}{Z_N}\det\left[y_i^{\frac{j-1}{2}}I_{j-1}\left(2\delta(\mu)\sqrt{y_i}\right)\right]_{i,j=1}^N
\det\left[y_i^{\frac{j+\nu-1}{2}}K_{j+\nu-1}\left(2\alpha(\mu)\sqrt{y_i}\right)\right]_{i,j=1}^N,
\end{split}
\end{equation}
where
\begin{equation}\label{ZN}
Z_N=\frac{N!\,
\alpha(\mu)^{N\nu+\frac{N(N-1)}{2}}\delta(\mu)^{\frac{N(N-1)}{2}}}{2^N\left(\alpha(\mu)^2-\delta(\mu)^2\right)^{N\nu+N^2}}\;\prod\limits_{j=1}^N\Gamma(j)\Gamma(j+\nu).
\end{equation}
\end{thm}
Let us regard $\mu$ as a deformation parameter, and consider two interesting limits of the joint probability density function  $P(y_1,\ldots,y_N)$.
In the first limiting case the two Gaussian matrices become independent,
corresponding to $\mu\rightarrow 1$, with  
$\delta(\mu)\rightarrow 0$ and $\alpha(\mu)\rightarrow 1$.
This fact is obvious from the very definition of two $\mu$-dependent Gaussian complex matrices.
It can also be seen directly from the explicit formula for the joint probability density function  $P(y_1,\ldots,y_N)$,
equation (\ref{TheMainJointProbabilityDensityFunction}), as shown in Appendix \ref{limitA}:
\begin{equation}\label{FirstLimtingRelation}
\underset{\mu\rightarrow 1}{\lim}
P(y_1,\ldots,y_N)
=\frac{\det\left[y_i^{j-1}\right]_{i,j=1}^N
}{N!\prod\limits_{j=1}^N\Gamma(j)^2\Gamma(j+\nu)}
\det\left[G^{2,0}_{0,2}\left(\begin{array}{cc}
                - \\
               0,j+\nu-1
             \end{array}\biggl|y_i
\right)\right]_{i,j=1}^N.
\end{equation}
Here we have introduced the Meijer G-function (see e.g. Luke \cite{Luke})
\begin{equation}
G_{p,q}^{m,n}\left(\begin{array}{cccc}
                                                                a_1, & a_2, & \ldots, & a_p \\
                                                                b_1, & b_2, & \ldots, & b_q
                                                              \end{array}\biggl| z
\right)=\frac{1}{2\pi i}\int\limits_{C}\frac{\prod_{j=1}^m\Gamma(b_j-s)\prod_{j=1}^n\Gamma(1-a_j+s)}{\prod_{j=m+1}^q\Gamma(1-b_j+s)\prod_{j=n+1}^p\Gamma(a_j-s)}z^sds.
\label{G-def}
\end{equation}
An empty product is interpreted as unity, for the indices $m\geq q$, $n\geq p$. The contour of integration $C$ depends on the location of the poles of the Gamma functions, and we refer to the NIST handbook \cite{NIST} for details on the different possibilities. 
In particular  the following formula holds, see 9.34.3 in \cite{GradshteinRyzhik},
\begin{equation}\label{K=G}
G^{2,0}_{0,2}\left(\begin{array}{cc}
                - \\
               0,l
             \end{array}\biggl|y
\right)=2y^{\frac{l}{2}}K_l(2\sqrt{y}) .
\end{equation}
The right-hand side of equation (\ref{FirstLimtingRelation}) 
agrees with the joint probability density function  
of squared singular values
of two independent rectangular complex Ginibre matrices, see Akemann, Ipsen and Kieburg 
\cite{AkemannIpsenKieburg}, formulae (18) and (21).
Here we only consider the special case that the matrix $Y=X_1X_2$ is square, with $\nu=\nu_1$ and $\nu_2=0$ compared to there. We will need $Y$ to be square for the group integrals that we encounter in the derivation for general $\mu\in(0,1)$.

The second interesting limit is that of $\mu\rightarrow 0$. In this limit 
$\delta(\mu)$ and  $\alpha(\mu)$ diverge, 
and we obtain the joint density  equivalent to the classical Laguerre ensemble.
To find the limit of the joint probability density function as $\mu\rightarrow 0$ we use formula
(\ref{TheMainJointProbabilityDensityFunction}), and replace the modified Bessel functions inside the determinants by their large argument asymptotic expressions. A short calculation in Appendix \ref{limitA} yields 
\begin{equation}\label{SecondLimtingRelation}
\underset{\mu\rightarrow 0}{\lim}P(y_1,\ldots,y_N)
=\frac{2^{N(M-1)}}{N!\prod\limits_{j=1}^N\Gamma(j)\Gamma(j+\nu)}
\left(\det\left[y_i^{\frac{j-1}{2}}\right]_{i,j=1}^N\right)^2
\prod\limits_{i=1}^Ny_i^{\frac{\nu-1}{2}}\exp\left[-2y_i^{\frac{1}{2}}\right].
\end{equation}
Changing variables in equation (\ref{SecondLimtingRelation}),
\begin{equation}\label{change}
y_i\mapsto v_i=2y_i^{\frac{1}{2}},
\end{equation}
we obtain the joint probability density function of the classical Laguerre ensemble
$$
\frac{1}{N!\prod\limits_{j=1}^N\Gamma(j)\Gamma(j+\nu)}
\left(\det\left[v_i^{j-1}\right]_{i,j=1}^N\right)^2\prod\limits_{i=1}^Nv_i^\nu e^{-v_i},
$$
see Forrester \cite{ForresterLogGases}, Chapter 7. The change of variables is necessary because we started from the singular values of $Y=X_1X_1^*$ in this limit,  rather than of $X_1$ which is the single matrix with Gaussian distribution left in this limit.

We conclude that the 
product $X_1X_2$ of two $\mu$-dependent Gaussian complex matrices
represents an \textit{interpolating biorthogonal ensemble}.
It interpolates between the ensemble 
of two independent complex Gaussian matrices, and the Laguerre
ensemble of a single complex Gaussian matrix.

\subsection{Exact formulae for the correlation kernel}
Theorem \ref{PropositionTheMainJointProbabilityDensity} implies that the squared singular values $y_1$, $\ldots$, $y_N$ of
the product $X_1X_2$ of two $\mu$-dependent Gaussian complex matrices form a determinantal point process,
\begin{equation}\label{detPP}
P(y_1,\ldots,y_N)=
 \det\left[K_N(y_i,y_j)\right]_{i,j=1}^N\ .
\end{equation}
Here we present exact formulae for the correlation kernel of this process.
\begin{thm}\label{TheoremCorrelationKernelExact1}
The correlation kernel $K_N(x,y)$ of the determinantal point process formed by the squared singular values of  $X_1X_2$ is given by
\begin{equation}\label{K}
K_N(x,y)=\sum\limits_{n=0}^{N-1}P_n(x)Q_n(y),
\end{equation}
where the functions $P_{0}(x)$, $P_{1}(x)$, $\ldots$ are defined by
 \begin{equation}\label{P}
P_n(x)=(-1)^n(\nu+n)!n!\sum\limits_{k=0}^n\frac{\left(\alpha(\mu)^2-\delta(\mu)^2\right)^{k+\frac{1}{2}}}{\delta(\mu)^k}\frac{(-n)_k}{(\nu+k)!k!}
\,x^{\frac{k}{2}}I_k(2\delta(\mu)\sqrt{x}),
\end{equation}
and the functions
$Q_{0}(y)$, $Q_{1}(y)$, $\ldots$ are defined by
\begin{equation}\label{Q}
Q_{n}(y)=(-1)^n\frac{2}{(n!)^2}
\sum\limits_{l=0}^{n}\frac{\left(\alpha(\mu)^2-\delta(\mu)^2\right)^{l+\nu+\frac{1}{2}}}{\alpha(\mu)^{l+\nu}}
\frac{(-n)_l}{(\nu+l)!l!}\,y^{\frac{l+\nu}{2}}K_{l+\nu}(2\alpha(\mu)\sqrt{y}).
\end{equation}
\end{thm}
Using explicit formulae for the functions $P_n(x)$ and $Q_n(x)$ (equations (\ref{P}) and (\ref{Q})) we derive the following formula for the correlation kernel $K_N(x,y)$
\begin{thm}\label{TheoremBorodinTypeFormula} The correlation kernel $K_{N}(x,y)$  can written as
\begin{equation}\label{KNMExplicitFormula1}
\begin{split}
K_{N}(x,y)=&\ 2\sum\limits_{k,l=0}^{N-1}\sum\limits_{i=0}^l
\frac{(-1)^{i+k}(\nu+N+i)!}{(N-1-k)!(\nu+k)!i!(l-i)!k!(\nu+i)!(\nu+k+i+1)}
\\
&\times\frac{\left(\alpha(\mu)^2-\delta(\mu)^2\right)^{k+l+\nu+1}}{\alpha(\mu)^{\nu+l}
\delta(\mu)^{k}}{x^{\frac{k}{2}}y^{\frac{l+\nu}{2}}I_k(2\delta(\mu)\sqrt{x})K_{l+\nu} (2\alpha(\mu)\sqrt{y})}\ .
\end{split}
\end{equation}
\end{thm}
Theorem \ref{TheoremBorodinTypeFormula} enables us to compare our biorthogonal ensemble with the family of the Laguerre-type biorthogonal ensembles introduced and studied
in Borodin \cite{BorodinBO}, Section 4. Also,  Theorem \ref{TheoremBorodinTypeFormula} can be used to investigate the transition of our biorthogonal ensemble to a Laguerre-type ensemble as $\mu$ approaches zero.
Consider the Laguerre-type ensemble defined by the right-hand side of equation (\ref{SecondLimtingRelation}). Using the same argument  as in Borodin \cite{BorodinBO},
Section 4, Theorem 4.1 we can write the correlation kernel $K^{\rm Lag}_N(x,y)$ of the ensemble from \cite{BorodinBO} with $\theta=1$ as
\begin{equation}\label{KNLaguerreExplicitFormula1}
K_{N}^{\rm Lag}(x,y)=\frac{e^{-x^{\frac{1}{2}}-y^{\frac{1}{2}}}}{x^{\frac{1}{4}}y^{\frac{1}{4}}}\sum\limits_{k,l=0}^{N-1}\sum\limits_{i=0}^l
\frac{(-1)^{i+k}(\nu+N+i)!2^{k+l+\nu}}{(N-1-k)!(\nu+k)!i!(l-i)!k!(\nu+i)!}
\frac{x^{\frac{k}{2}}y^{\frac{l+\nu}{2}}}{(\nu+k+i+1)}.
\\
\end{equation}
It is not hard to check using the asymptotic expressions for the modified Bessel functions of large arguments (see eq. (\ref{Bessasypm})) that the kernel $K_N(x,y)$ turns into a kernel  equivalent to
$K_{N}^{\rm Lag}(x,y)$ as $\mu\rightarrow 0$.\footnote{Two kernels $K(x,y)$ and $K'(x,y)$ are called equivalent
if $\det\left[K(x_i,x_j]\right)_{i,j=1}^m=\det\left[K'(x_i,x_j)\right]_{i,j=1}^m$, for any $m=1,2,\ldots$, for example $K'(x,y)=(f(x)/f(y))K(x,y)$. Thus two equivalent kernels define the same correlation functions.}\\

The subsequent asymptotic analysis requires a detailed investigation of the properties of the functions $P_n(x)$ and $Q_n(y)$ determining the correlation
kernel $K_N(x,y)$. In particular, we show that these functions satisfy the following biorthogonality condition.
\begin{prop}\label{Biorthogonality}
The functions  $P_n(x)$, $Q_n(x)$ defined by equations (\ref{P}) and (\ref{Q}) correspondingly satisfy the biorthogonality condition:
$$
\int\limits_{0}^{\infty}P_n(x)Q_m(x)dx=\delta_{n,m},\;\; n,m=0,1,2,\ldots .
$$
\end{prop}
Moreover,  both $P_n(x)$ and $Q_n(y)$ satisfy
five term recurrence relations, and can be represented as contour integrals.
Namely, the following Proposition holds true.
\begin{prop}\label{PropositionRecurrenceRelations}
(a) For the functions $P_n(x)$ we have the following five term recurrence relation
\begin{equation}\label{xPn1}
xP_n(x)=a_{2,n}P_{n+2}(x)+a_{1,n}P_{n+1}(x)+a_{0,n}P_{n}(x)
+a_{-1,n}P_{n-1}(x)+a_{-2,n}P_{n-2}(x),
\end{equation}
where the coefficients $a_{2,n}$, $a_{1,n}$, $a_{0,n}$, $a_{-1,n}$, and $a_{-2,n}$
are given explicitly by
\begin{eqnarray}
\ \ \ \ \ \ \ \ a_{2,n}&=&\frac{1}{(n+2)(n+1)}\frac{\delta(\mu)^2}{\left(\alpha(\mu)^2-\delta(\mu)^2\right)^2},\\
a_{1,n}&=&\frac{1}{\alpha(\mu)^2-\delta(\mu)^2}+\frac{2(2n+\nu+2)}{(n+1)}\frac{\delta(\mu)^2}{\left(\alpha(\mu)^2-\delta(\mu)^2\right)^2},\\
a_{0,n}&=&\frac{3n^2+2\nu n+3n+\nu+1}{\alpha(\mu)^2-\delta(\mu)^2}+(6n^2+6n\nu+\nu^2+6n+3\nu+2)\frac{\delta(\mu)^2}{\left(\alpha(\mu)^2-\delta(\mu)^2\right)^2},\\
a_{-1,n}&=&\frac{n^2(n+\nu)(3n+\nu)}{\alpha(\mu)^2-\delta(\mu)^2}+2n^2(\nu+n)(2n+\nu)\frac{\delta(\mu)^2}{\left(\alpha(\mu)^2-\delta(\mu)^2\right)^2},\\
\label{xPn6}
a_{-2,n}&=&
(\nu+n)(\nu+n-1)n^2(n-1)^2\frac{\alpha(\mu)^2}{\left(\alpha(\mu)^2-\delta(\mu)^2\right)^2}.
\end{eqnarray}
(b) For the functions $Q_n(y)$ we have the following five term recurrence relation
\begin{equation}\label{yQn1}
yQ_n(y)=b_{2,n}Q_{n+2}(y)+b_{1,n}Q_{n+1}(y)+b_{0,n}Q_{n}(y)
+b_{-1,n}Q_{n-1}(y)+b_{-2,n}Q_{n-2}(y),
\end{equation}
where the coefficients $b_{2,n}$, $b_{1,n}$, $b_{0,n}$, $b_{-1,n}$, and $b_{-2,n}$ are given explicitly by
\begin{eqnarray}
\ \ \ \ \ \ \ \ b_{2,n}&=&(\nu+n+2)(\nu+n+1)(n+2)^2(n+1)^2\frac{\alpha(\mu)^2}{\left(\alpha(\mu)^2-\delta(\mu)^2\right)^2},\\
b_{1,n}&=&-\frac{(n+1)^2(n+\nu+1)^2}{\alpha(\mu)^2-\delta(\mu)^2}
+2(2n+\nu+2)(n+\nu+1)(n+1)^2\frac{\alpha(\mu)^2}{\left(\alpha(\mu)^2-\delta(\mu)^2\right)^2},\\
b_{0,n}&=&-\frac{(n+\nu)^2+2(n+1)(n+\nu)+n+1}{\alpha(\mu)^2-\delta(\mu)^2}\\
&&+\left((n+\nu)(5n+\nu+3)+n(n+3)+2\right)\frac{\alpha(\mu)^2}{\left(\alpha(\mu)^2-\delta(\mu)^2\right)^2},\nonumber\\
b_{-1,n}&=&-\frac{(3n+2\nu)}{n}\frac{1}{\alpha(\mu)^2-\delta(\mu)^2}
+\frac{2(2n+\nu)}{n}\frac{\alpha(\mu)^2}{\left(\alpha(\mu)^2-\delta(\mu)^2\right)^2},\\
b_{-2,n}&=&
\frac{1}{n(n-1)}\frac{\delta(\mu)^2}{\left(\alpha(\mu)^2-\delta(\mu)^2\right)^2}.
\label{yQn6}
\end{eqnarray}
\end{prop}
Note that
the recurrence coefficients are related as
\begin{equation}\label{RelationRecurrenceCoefficients}
a_{2,n}=b_{-2,n+2},\;a_{1,n}=b_{-1,n+1},\;a_{0,n}=b_{0,n},\;a_{-1,n}=b_{1,n-1},\;a_{-2,n}=b_{2,n-2}.
\end{equation}
This follows from  the biorthogonality of the functions $P_n(x)$ and $Q_n(y)$, see Proposition \ref{Biorthogonality}.
Equation (\ref{RelationRecurrenceCoefficients})
can be  checked directly as well using the formulas in Proposition \ref{PropositionRecurrenceRelations} for the recurrence coefficients.

Using the recurrence relations stated in Proposition \ref{PropositionRecurrenceRelations} we derive the following Christoffel-Darboux
type formula for the correlation kernel $K_N(x,y)$.
\begin{thm}\label{TheoremChristoffelDarboux}
The Christoffel-Darboux type formula for the correlation kernel $K_{N}(x,y)$ valid for $N\geq2$ and $x\neq y$ is given by 
\begin{equation}\label{CDKernel}
\begin{split}
K_{N}(x,y)=-\frac{a_{-2,N}P_{N-2}(x)Q_{N}(y)+a_{-2,N+1}P_{N-1}(x)Q_{N+1}(y)+a_{-1,N}P_{N-1}(x)Q_{N}(y)}{x-y}\\
+\frac{a_{1,N-1}P_{N}(x)Q_{N-1}(y)+a_{2,N-2}P_{N}(x)Q_{N-2}(y)+a_{2,N-1}P_{N+1}(x)Q_{N-1}(y)}{x-y},
\end{split}
\end{equation}
where the coefficients $a_{-2,N}$, $a_{-1,N}$, $a_{1,N}$ and $a_{2,N}$ are given by Proposition \ref{PropositionRecurrenceRelations}.
\end{thm}

The next Proposition gives contour integral representations for the functions $P_n(x)$ and $Q_n(y)$.
\begin{prop}\label{PropositionContourIntegralRepresentationsPQ}
(a) The following contour integral representation for the function $P_n(x)$ holds:
\begin{equation}\label{ContourIntegralRepresentationsPn}
\begin{split}
P_n(x)=&\frac{1}{2\pi i}(\nu+n)!(n!)^2\left(\alpha(\mu)^2-\delta(\mu)^2\right)^{\frac{1}{2}}\\
&\times\oint\limits_{\Sigma}\frac{\Gamma(t-n)\left(\alpha(\mu)^2-\delta(\mu)^2\right)^tx^t}{\left(\Gamma(t+1)\right)^2\Gamma(t+\nu+1)}
{}_0F_1\left(\begin{array}{c}
                - \\
               t+1
             \end{array}\biggl|\delta(\mu)^2x
\right)dt,
\end{split}
\end{equation}
where $\Sigma$ is a closed contour that encircles $0$, $1$, $\ldots$, $n$ once in positive direction, $n=0,1,\ldots$, and $x>0$.\\
(b) The following contour integral representation for the function $Q_n(y)$ is true:
\begin{equation}\label{ContourIntegralRepresentationsQn}
\begin{split}
Q_n(y)=&\frac{1}{2\pi i(n!)^2(n+\nu)!}\left(1-\frac{\delta(\mu)^2}{\alpha(\mu)^2}\right)^{\nu}\left(\alpha(\mu)^2-\delta(\mu)^2\right)^{\frac{1}{2}}\\
&\times\int\limits_{c-i\infty}^{c+i\infty}\frac{\Gamma^2(s)\Gamma(s+\nu)}{\Gamma(s-n)}
{}_2F_1\left(\begin{array}{c}
               -n,  \nu+s \\
               s-n
             \end{array}
             \biggl|\frac{\delta(\mu)^2}{\alpha(\mu)^2}\right)\left(\alpha(\mu)^2y\right)^{-s}ds,
\end{split}
\end{equation}
where $c>0$, $n=0,1,\ldots$, and $y>0$.
\end{prop}
Finally, we  state  that as a consequence the correlation kernel $K_N(x,y)$ admits a double contour integral representation.
\begin{thm}\label{TheoremKContour}The correlation kernel, $K_{N}(x,y)$, can be written as
\begin{equation}
K_{N}(x,y)=\sum\limits_{k=0}^{N-1}K_{N}^{(k)}(x,y)\left(\frac{\delta(\mu)}{\alpha(\mu)}\right)^k,
\end{equation}
where
\begin{equation}\label{KorrelationKernelContourIntegral}
\begin{split}
K_{N}^{(k)}(x,y)=&
\sum\limits_{m=0}^k\frac{(-1)^m}{(2\pi i)^2}\left(\begin{array}{c}
                                        N \\
                                        m
                                      \end{array}
\right)\oint\limits_{\Sigma}dt\int\limits_{c-i\infty}^{c+i\infty}ds\frac{\Gamma^2(s)\Gamma(s+\nu+k)\Gamma(s-t+m-1)\Gamma(t-N+1)}{\Gamma^2(t+1)\Gamma(t+\nu+1)
\Gamma(s-N+m)\Gamma(s-t+k)}\\
&\times\left(1-\frac{\delta(\mu)^2}{\alpha(\mu)^2}\right)^{\nu}
\left(\alpha(\mu)^2-\delta(\mu)^2\right)^{t+1} x^t\ 
{}_0F_1\left(\begin{array}{c}
                - \\
               t+1
             \end{array}\biggl|\delta(\mu)^2x
\right)\left(\alpha(\mu)^2y\right)^{-s}.
\end{split}
\end{equation}
The contour $\Sigma$ is chosen in the same way as in Proposition \ref{PropositionContourIntegralRepresentationsPQ}, and $c>0$.
\end{thm}
As $\mu\rightarrow 1$, the biorthogonal ensemble defined by equation (\ref{TheMainJointProbabilityDensityFunction}) turns into
that for the squared singular values of the product of  two matrices with independent complex Gaussian entries, see equation (\ref{FirstLimtingRelation}).
The biorthogonal ensemble for the squared singular values of products of $M$ matrices with independent complex Gaussian entries was
studied in \cite{AkemannKieburgWei,AkemannIpsenKieburg}. As $\mu\rightarrow 1$, the functions $P_n(x)$ and $Q_n(y)$ defined by equations
(\ref{P}) and (\ref{Q}) turn into the biorthogonal polynomials and their normalised dual functions there, see equations (43) and (47) in  \cite{AkemannIpsenKieburg}, resepctively.
Furthermore, our Propositions \ref{PropositionRecurrenceRelations} and \ref{PropositionContourIntegralRepresentationsPQ} are extensions
of the results obtained by  Kuijlaars and Zhang, see their Proposition 3.2 and formula (3.6), and the recurrence relations in Section 4 of Kuijlaars and Zhang \cite{KuijlaarsZhang}.
As $\mu\rightarrow 1$, the formulae for the correlation kernel $K_N(x,y)$ given in Theorem \ref{TheoremKContour} turn into the double integral formula of Proposition 5.1
in Kuijlaars and Zhang \cite{KuijlaarsZhang}.
\subsection{The hard edge scaling limit of the correlation kernel}
We use the Christoffel-Darboux type formula for the correlation kernel $K_N(x,y)$ given by Theorem  \ref{TheoremChristoffelDarboux},
and the contour integral representations for the functions $P_n(x)$ and $Q_n(y)$ of Proposition \ref{PropositionContourIntegralRepresentationsPQ}
to find the scaling limit of $K_N(x,y)$ near the origin (hard edge).
\begin{thm}\label{TheoremHardEdgeScalingLimit}
Let $\nu$ and $\mu$ be fixed. For $x$ and $y$ in a compact subset of the positive real axis,
\begin{equation}
K_{\nu}(x,y)=\underset{N\rightarrow\infty}{\lim}\left\{\frac{1}{N\left(\alpha(\mu)^2-\delta(\mu)^2\right)}
K_{N}\left(\frac{x}{N\left(\alpha(\mu)^2-\delta(\mu)^2\right)},
\frac{y}{N\left(\alpha(\mu)^2-\delta(\mu)^2\right)}\right)\right\},
\nonumber
\end{equation}
where the limiting Meijer G-kernel $K_{\nu}(x,y)$ is given by
$$
K_{\nu}(x,y)=\int\limits_0^1G^{1,0}_{0,3}\left(\begin{array}{ccc}
                       & - &  \\
                      0, & -\nu, & 0
                    \end{array}
\biggr|ux\right)G^{2,0}_{0,3}\left(\begin{array}{ccc}
                       & - &  \\
                      \nu,&0, & 0
                    \end{array}
\biggr|uy\right)du.
$$
Here $G^{1,0}_{0,3}\left(\begin{array}{ccc}
                       & - &  \\
                      0, & -\nu, & 0
                    \end{array}
\biggr|ux\right)$ and $G^{2,0}_{0,3}\left(\begin{array}{ccc}
                       & - &  \\
                      \nu,&0, & 0
                    \end{array}
\biggr|uy\right)$ are Meijer $G$-functions with a suitable choice of parameters.
\end{thm}
The resulting limiting kernel $K_{\nu}(x,y)$ coincides  with the scaling limit found by
Bertola, Gekhtman, and Szmigielski in the Cauchy-Laguerre two-matrix model \cite{Bertola},  with the scaling limit for the product of two independent complex Gaussian
matrices found by Kuijlaars and Zhang \cite{KuijlaarsZhang}, and with the limiting kernel for the product of two independent complex Gaussian
matrices times a fixed arbitrary number of inverses of such matrices found by Forrester \cite{ForresterProductWishart}.
This confirms that the family of new limiting so-called Meijer G-kernels obtained in Kuijlaars and Zhang \cite{KuijlaarsZhang} in the context of products of independent matrices
represents a new universality class.
\subsection{The Central Limit Theorem}
Proposition \ref{PropositionRecurrenceRelations} gives explicitly the recurrence coefficients for the functions $P_n(x)$
and $Q_n(y)$ determining the correlation kernel of the biorthogonal ensemble defined by equation (\ref{TheMainJointProbabilityDensityFunction}).
This enables us to derive a Central Limit Theorem for the linear statistics of singular values of $X_1X_2$, and to give the limiting variance explicitly.

Here, instead of the probability distribution $P_{N,M}(X_1,X_2)dX_1dX_2$ (defined by equation (\ref{JointDensityX1X2}))
we consider the probability distribution $\widetilde{P_{N,M}}(X_1,X_2)dX_1dX_2$ on the Cartesian
product of $\Mat(\C, N\times  M)$ and $\Mat(\C, M\times  N)$ defined by
\begin{equation}\label{PNEW}
\begin{split}
\widetilde{P_{N,M}}(X_1,X_2)dX_1dX_2=&c\cdot \exp\left[-N\alpha(\mu)\Tr(X_1X_1^*+X_2^*X_2)+N\delta(\mu)\Tr(X_1X_2+X_2^*X_1^*)\right]\\
&\times\prod\limits_{i=1}^N\prod\limits_{j=1}^Md{X_{i,j}^{(1)}}^Rd{X_{i,j}^{(1)}}^I\prod\limits_{i=1}^M\prod\limits_{j=1}^Nd{X_{i,j}^{(2)}}^Rd{X_{i,j}^{(2)}}^I,
\end{split}
\end{equation}
where $X_{i,j}^{(1)}={X_{i,j}^{(1)}}^R+i{X_{i,j}^{(1)}}^I$ and $X_{i,j}^{(2)}={X_{i,j}^{(2)}}^R+i{X_{i,j}^{(2)}}^I$ denote the sums of the real and imaginary parts of the matrix entries $X_{i,j}^{(1)}$ and $X_{i,j}^{(2)}$, and $c$ is a normalising constant. Equation (\ref{PNEW}) is obtained from Equation (\ref{JointDensityX1X2}) by a simple rescaling of the matrix elements by $\sqrt{N}$.

Let $y_1$, $\ldots$, $y_N$ be the squared singular values of the matrix $X_1X_2$,
and define the linear statistics of 
$y_1$, $\ldots$, $y_N$ by the formula
$$
Y_f^{(N)}=\sum\limits_{i=1}^Nf(y_i).
$$
\begin{thm}\label{CentralLimitTheorem} Let $f$ be a polynomial with real coefficients. Then we have
$$
Y_f^{(N)}-\mathbb{E}Y_f^{(N)}\rightarrow {\mathcal N}\left(0,\sum\limits_{k=1}^{\infty}k\hat{f}_k\hat{f}_{-k}\right)
$$
in distribution, where
\begin{eqnarray}
\hat{f}_k&=&\frac{1}{2\pi i}\oint_{|w|=1}f\left(s(w;\mu)
\right)w^k\frac{dw}{w}\ ,\nonumber\\
s(w;\mu)&=&\frac{1}{4w^2}(w+1)^3\left( w(1-\mu)^2+(1+\mu)^2\right).
\label{sw}
\end{eqnarray}
\end{thm}
\begin{rem}
Since $f$ is a polynomial with real coefficients, $\hat{f}_k$ is real. Furthermore, the Central Limit Theorems for the limiting cases $\mu\to0$ and $\mu\to1$ can be immediately read off by taking the limits on the Laurent polynomial $s(w;\mu)$. For the product of two independent complex Gaussian matrices we obtain $\lim_{\mu\to1}s(w;\mu)={(w+1)^3}/{w^2}$. 
This agrees with the results following from the recursion coefficients by Kuijlaars and Zhang, by specifying to two matrices there. 
In the opposite limit we obtain $\lim_{\mu\to0}s(w;\mu)={(w+1)^4}/{(4w^2)}$.  It is not difficult using Laguerre polynomials of square root arguments to directly show that this is the correct limit for the ensemble in eq. (\ref{SecondLimtingRelation}) - which is not the standard Laguerre ensemble due to the change of variables in equation (\ref{change}).
\end{rem}
The proof of Theorem \ref{CentralLimitTheorem} uses the results  for biorthogonal ensembles obtained in Breuer and Duits \cite{BreuerDuits}.
They showed that whenever the asymptotic of recurrence coefficients is available, a Central Limit Theorem
for the linear statistics can be derived. In our case, Proposition \ref{PropositionRecurrenceRelations} gives the recurrence coefficients explicitly.
Considering the rescaled probability distribution $\widetilde{P_{N,M}}(X_1,X_2)dX_1dX_2$ 
we obtain recurrence
coefficients that have finite limits as $N\rightarrow\infty$, which gives Theorem \ref{CentralLimitTheorem}.

\section{Proof of Theorem \ref{PropositionTheMainJointProbabilityDensity}}\label{ProofsInitial}
First we show that the computation of the joint probability density function of (squared) singular values
for products of rectangular matrices can be reduced to that for products of square matrices of the same size.
Namely, the following Lemma holds true.
\begin{lem}\label{LemmaReduction}
Let $X\in\Mat(\C, M\times N)$ and $G\in\Mat(\C, N\times M)$ be two $\mu$-dependent Gaussian complex matrices.
Assume that $M\geq N$. Then the squared singular values of the matrix $GX$ are distributed in the same way as the squared singular values
of the matrix $G_0X_0$, where  $G_0\in\Mat(\C, N\times N)$, $X_0\in\Mat(\C, N\times N)$, and the joint distribution of $G_0$, $X_0$ is given by
\begin{equation}\label{G0X0Distribution}
\begin{split}
P^{(N,M)}(G_0,X_0)dG_0dX_0=&\const\cdot\det\left(X_0^*X_0\right)^{M-N}\\
&\times e^{-\alpha(\mu)\Tr(G_0G_0^*+X_0^*X_0)+\delta(\mu)\Tr(G_0X_0+X_0^*G_0^*)}dG_0dX_0.
\end{split}
\end{equation}
\end{lem}
Here and below the computation of the $\mu$-dependent constants is suppressed until the last part of the proof of Theorem \ref{PropositionTheMainJointProbabilityDensity}.
\begin{proof}
If $M=N$, then the statement of the Lemma follows immediately. Consider the case when $M>N$.
Recall that the matrices $G$, $X$  are distributed in accordance with
\begin{equation}
\begin{split}
P^{(N,M)}(G,X)dGdX=\const\cdot e^{-\alpha(\mu)\Tr(GG^*+X^*X)+\delta(\mu)\Tr(GX+X^*G^*)}dGdX.
\end{split}
\nonumber
\end{equation}
Consider the following decomposition of the matrix $X$
$$
X=U\left(\begin{array}{c}
           X_0 \\
           O_{M-N,N}
         \end{array}
\right),
$$
where $U$ is an $M\times M$ unitary matrix, $X_0$ is an $N\times N$ complex matrix, and $O_{M-N,N}$ is a complex matrix of size $(M-N)\times N$ with zero entries.\footnote{For a proof of the existence of such a decomposition see Fischmann, Bruzda, Khoruzhenko, Sommers, and Zyczkowski
\cite{Fischmann}, Section 2.}
We have
\begin{equation}
\begin{split}
&P^{(N,M)}(G,X)dGdX=\const\cdot\det\left(X_0^*X_0\right)^{M-N}
e^{-\alpha(\mu)\Tr(GG^*+X_0^*X_0)}\\
&\ \ \ \ \times \exp\left[\delta(\mu)\left\{\Tr\left(GU\left(\begin{array}{c}
           X_0 \\
           O_{M-N,N}
         \end{array}
\right)\right)+\Tr\left(\left(\begin{array}{cc}
                                      X_0^* & O_{N,M-N}
                                    \end{array}
\right)U^*G^*\right)\right\}\right]dGdUdX_0,
\end{split}
\nonumber
\end{equation}
where  we have used the results of Section 2 in  Fischmann, Bruzda, Khoruzhenko, Sommers, and Zyczkowski
\cite{Fischmann} (see also the discussion in Ipsen and Kieburg \cite{IpsenKieburg}, Section III, A).
Here $dU$ denotes the Haar measure.
If $\hat{G}=GU$, then the equation above can be rewritten as
\begin{equation}\label{GX1}
\begin{split}
&P^{(N,M)}(G,X)dGdX=\const\cdot\det\left(X_0^*X_0\right)^{M-N}
e^{-\alpha(\mu)\Tr(\hat{G}\hat{G}^*+X_0^*X_0)}\\
&\times \exp\left[\delta(\mu)\left\{\Tr\left(\hat{G}\left(\begin{array}{c}
           X_0 \\
           O_{M-N,N}
         \end{array}
\right)\right)+\Tr\left(\left(\begin{array}{cc}
                                      X_0^* & O_{N,M-N}
                                    \end{array}
\right)\hat{G}^*\right)\right\}\right]d\hat{G}dUdX_0,
\end{split}
\end{equation}
where we have used the invariance of the corresponding Lebesgue measure $dG$ under unitary transformations.
Now, set
\begin{equation}\label{HatG}
\hat{G}=\left(\begin{array}{cc}
                G_0 & \hat{G}_{N,M-N}
              \end{array}
\right).
\end{equation}
This is a block decomposition of the rectangular matrix $\hat{G}$ of size $N\times M$ $(M>N)$ such that $G_0$ is the square matrix of size $N\times N$
whose entries are those of the first $N$ columns of $\hat{G}$, and $\hat{G}_{N,M-N}$ is the remaining rectangular matrix of size $N\times (M-N)$.
Inserting (\ref{HatG}) into equation (\ref{GX1}), we obtain
\begin{equation}
\begin{split}
P^{(N,M)}(G,X)dGdX=&\const\cdot\det\left(X_0^*X_0\right)^{M-N}
e^{-\alpha(\mu)\left(\Tr\left(G_0G_0^*\right)+\Tr\left(\hat{G}_{N,M-N}\hat{G}_{N,M-N}^*\right)+\Tr\left(X_0^*X_0\right)\right)}\\
&\times e^{\delta(\mu)\Tr\left(G_0X_0+X_0^*G_0^*\right)}d\hat{G}_0
d\hat{G}_{N,M-N}dUdX_0.
\end{split}
\nonumber
\end{equation}
The formula just written above implies that the joint distribution of $G_0$, $X_0$ is given by equation (\ref{G0X0Distribution}).
Moreover, by construction the squared singular values of $GX$ coincide with those of $G_0X_0$.
\end{proof}
Let us turn to the proof of Theorem \ref{PropositionTheMainJointProbabilityDensity}. We use Lemma \ref{LemmaReduction}, and
 assume that both matrices $X_1$, $X_2$ are taken from $\Mat\left(\C,N\times N\right)$, and that the joint distribution of $X_1$, $X_2$ is given by
\begin{equation}\label{GXsquareDistribution}
\begin{split}
P^{(N,M)}(X_1,X_2)dX_1dX_2=&\const\cdot\det\left(X_2^*X_2\right)^{M-N}\\
&\times e^{-\alpha(\mu)\Tr(X_1X_1^*+X_2^*X_2)+\delta(\mu)\Tr(X_1X_2+X_2^*X_1^*)}dX_1dX_2,
\end{split}
\end{equation}
where $M\geq N$. In fact we need that $X_1,X_2\in Gl(N,\C)$. Because the set of invertible matrices is dense in $\Mat\left(\C,N\times N\right)$ this will not change the joint distribution.
Consider the change of variables
$$
X_1\mapsto Y_2=X_1X_2,\;\;X_2\mapsto Y_1=X_2.
$$
It is known that this transformation has a Jacobian $\det\left(Y_1^*Y_1\right)^{-N}$. Therefore we can write
\begin{equation}\label{GXsquareDistribution1}
\begin{split}
P^{(N,M)}(X_1,X_2)dX_1dX_2=&\const\cdot\det\left(Y_1^*Y_1\right)^{M-2N}\\
&\times e^{-\alpha(\mu)\left(\Tr\left(Y_2^*Y_2Y_1^{-1}\left(Y_1^*\right)^{-1}\right)+\Tr(Y_1^*Y_1)\right)+\delta(\mu)\Tr(Y_2+Y_2^*)}dY_2dY_1.
\end{split}
\end{equation}
Next we use the singular value decomposition for both $Y_2$ and $Y_1$
$$
Y_1=V_1\Lambda_1U_1,\;\; \Lambda_1=\left(\begin{array}{ccc}
                                           \lambda_1^{(1)} &  & 0 \\
                                            & \ddots &  \\
                                           0 &  & \lambda_N^{(1)}
                                         \end{array}
\right),
$$
$$
Y_2=V_2\Lambda_2U_2,\;\; \Lambda_2=\left(\begin{array}{ccc}
                                           \lambda_1^{(2)} &  & 0 \\
                                            & \ddots &  \\
                                           0 &  & \lambda_N^{(2)}
                                         \end{array}
\right),
$$
where $\Lambda_1$, $\Lambda_2$ are diagonal matrices with the singular values along the diagonals, and $V_1$, $V_2$, $U_1$ and $U_2$
are unitary $N\times N$ matrices.  It is known that
$$
dY_1=\const\cdot\triangle\left(\Lambda_1^2\right)^2\left(\prod\limits_{j=1}^N\lambda_j^{(1)}d\lambda_j^{(1)}\right)dU_1dV_1,
$$
$$
dY_2=\const\cdot\triangle\left(\Lambda_2^2\right)^2\left(\prod\limits_{j=1}^N\lambda_j^{(2)}d\lambda_j^{(2)}\right)dU_2dV_2,
$$
where we have introduced the Vandermonde determinant
$$
\triangle\left(\Lambda_1^2\right)=\prod\limits_{N\geq j>k\geq 1}\left(\left(\lambda_j^{(1)}\right)^2-\left(\lambda_k^{(1)}\right)^2\right),\;\;
\triangle\left(\Lambda_2^2\right)=\prod\limits_{N\geq j>k\geq 1}\left(\left(\lambda_j^{(2)}\right)^2-\left(\lambda_k^{(2)}\right)^2\right),
$$
and where $dU_1$, $dV_1$, $dU_2$, and $dV_2$ are the corresponding Haar measures on the unitary group $U(N)$.
Combining these formulae we obtain a probability measure
\begin{equation}
\begin{split}
&P^{(N,M)}(X_1,X_2)dX_1dX_2=\const\cdot e^{-\alpha(\mu)\left(\Tr(\Lambda_1^2)+\Tr\left(U_1U_2^*\Lambda_2^2U_2U_1^*\Lambda_1^{-2}\right)\right)
+\delta(\mu)\left(\Tr(V_2\Lambda_2U_2)+\Tr\left(U_2^*\Lambda_2V_2^*\right)\right)}\\
&\times\triangle\left(\Lambda_1^2\right)^2\triangle\left(\Lambda_2^2\right)^2{\det}^{M-2N}\left[\Lambda_1^2\right]
\left(\prod\limits_{j=1}^N\lambda_j^{(1)}d\lambda_j^{(1)}\right)\left(\prod\limits_{j=1}^N\lambda_j^{(2)}d\lambda_j^{(2)}\right)
dU_1dU_2dV_1dV_2.
\end{split}
\nonumber
\end{equation}
Using the invariance of the Haar measures under the subsequent shifts
$$
U_1\mapsto U_1U_2,\;\;\;\mbox{and}\;\;\; U_2\mapsto U_2V_2^*,
$$
and  integrating over $V_1$ and $V_2$ we obtain
\begin{equation}\label{GXsquareDistribution2}
\begin{split}
&P^{(N,M)}(X_1,X_2)dX_1dX_2=\const\cdot e^{-\alpha(\mu)\left(\Tr(\Lambda_1^2)+\Tr\left(U_1\Lambda_2^2U_1^*\Lambda_1^{-2}\right)\right)
+\delta(\mu)\left(\Tr(\Lambda_2U_2)+\Tr\left(U_2^*\Lambda_2\right)\right)}\\
&\times\triangle\left(\Lambda_1^2\right)^2\triangle\left(\Lambda_2^2\right)^2{\det}^{M-2N}\left[\Lambda_1^2\right]
\left(\prod\limits_{j=1}^N\lambda_j^{(1)}d\lambda_j^{(1)}\right)\left(\prod\limits_{j=1}^N\lambda_j^{(2)}d\lambda_j^{(2)}\right)
dU_1dU_2.
\end{split}
\end{equation}
The integration over $U_1$ can be performed using the Harish-Chandra-Itzykson-Zuber integration formula \cite{HC,IZ}
\begin{equation}\label{HCIZ}
\int_{U(N)} e^{-\Tr\left(U_1\Lambda_2^2U_1^*\Lambda_1^{-2}\right)}dU_1=
\const\cdot\frac{\det\left[\exp\left[-\left(\lambda_j^{(2)}\right)^2
\left(\lambda_i^{(1)}\right)^{-2}\right]\right]_{i,j=1}^N}{\triangle(\Lambda_2^2)\triangle(\Lambda_1^{-2})},
\end{equation}
where the constant does not depend on $\Lambda_1$ and $\Lambda_2$,  and we have used the transformation $\Lambda_1\mapsto\alpha(\mu)^{\frac{1}{2}}\Lambda_1$. In addition, we apply  the fact that the Vandermonde determinant of inverse powers is proportional to the Vandermonde determinant with positive powers, namely
\begin{equation}\label{VandermondeIdentity}
\triangle\left(\Lambda_1^{-2}\right)=\const\cdot\frac{\triangle(\Lambda_1^2)}{\det^{N-1}\left[\Lambda_1^2\right]}.
\end{equation}
As a result of application of formulae (\ref{HCIZ}), (\ref{VandermondeIdentity}) to probability measure (\ref{GXsquareDistribution2})
we have
\begin{equation}
\begin{split}
&P^{(N,M)}(X_1,X_2)dX_1dX_2=\const\cdot e^{-\alpha(\mu)^2\Tr(\Lambda_1^2)
+\delta(\mu)\left(\Tr(\Lambda_2U_2)+\Tr\left(U_2^*\Lambda_2\right)\right)}\triangle\left(\Lambda_1^2\right)\triangle\left(\Lambda_2^2\right)\\
&\times
\det\left[\exp\left[-\left(\lambda_j^{(2)}\right)^2
\left(\lambda_i^{(1)}\right)^{-2}\right]\right]_{i,j=1}^N
{\det}^{M-N-1}\left(\Lambda_1^2\right)
\left(\prod\limits_{j=1}^N\lambda_j^{(1)}d\lambda_j^{(1)}\right)\left(\prod\limits_{j=1}^N\lambda_j^{(2)}d\lambda_j^{(2)}\right)
dU_2.
\end{split}
\end{equation}
Now our task is to perform the integration over $U_2$. This can be done exploiting the following Leutwyler-Smilga integral formula \cite{LeutwylerSmilga}, see e.g.
\cite{Baha} for a derivation based on group characters,
\begin{equation}\label{IZBessel}
\int_{U(N)} e^{
\delta(\mu)\Tr\left(\Lambda_2\left(U_2+U_2^*\right)\right)}dU_2
=\const\cdot\frac{\det\left[\left(\lambda_j^{(2)}\right)^{i-1}I_{i-1}\left(2\delta(\mu)\lambda_j^{(2)}\right)\right]_{i,j=1}^N}{\triangle(\Lambda_2^2)}.
\end{equation}
Here $I_k(x)$ denotes
the modified Bessel function of the first  kind.
After the integration over $U_2$ we obtain the following probability distribution
\begin{equation}\label{GXsquareDistribution3}
\begin{split}
&P^{(N,M)}(X_1,X_2)dX_1dX_2=\const\cdot e^{-\alpha^2(\mu)\Tr(\Lambda_1^2)}
\triangle\left(\Lambda_1^2\right)\det\left[\left(\lambda_j^{(2)}\right)^{i-1}I_{i-1}\left(2\delta(\mu)\lambda_j^{(2)}\right)\right]_{i,j=1}^N
\\
&\times
\det\left[\exp\left[-\left(\lambda_j^{(2)}\right)^2
\left(\lambda_i^{(1)}\right)^{-2}\right]\right]_{i,j=1}^N
{\det}^{M-N-1}\left[\Lambda_1^2\right]
\left(\prod\limits_{j=1}^N\lambda_j^{(1)}d\lambda_j^{(1)}\right)\left(\prod\limits_{j=1}^N\lambda_j^{(2)}d\lambda_j^{(2)}\right).
\end{split}
\end{equation}
To get the induced probability distribution of the singular values $\lambda_1^{(2)}$,$\ldots$,$\lambda_N^{(2)}$ of the matrix $Y_2=X_1X_2$ we only
need to integrate the probability distribution (\ref{GXsquareDistribution3}) over the variables $\lambda_1^{(1)}$,$\ldots$,$\lambda_N^{(1)}$. The integral
over these variables is
\begin{equation}\label{GXsquareDistribution4}
\begin{split}
{\mathcal I}=&\int\det\left[\left(\lambda^{(1)}_i\right)^{2(j-1)}\right]_{i,j=1}^N
\det\left[\exp\left[-\left(\lambda_j^{(2)}\right)^2
\left(\lambda_i^{(1)}\right)^{-2}\right]\right]_{i,j=1}^N\\
&\times\prod\limits_{j=1}^Ne^{-\alpha(\mu)^2\left(\lambda_j^{(1)}\right)^2}\left(\lambda^{(1)}_j\right)^{2\nu-1}d\lambda_j^{(1)}.
\end{split}
\end{equation}
Applying the Andr$\acute{\mbox{e}}$ief integral identity valid for a set of integrable functions,
\begin{equation}\label{AndId}
\begin{split}
\int\det\left[\varphi_i(x_j)\right]_{i,j=1}^N\det\left[\psi_i(x_j)\right]_{i,j=1}^N\prod\limits_{j=1}^Nd\mu(x_j)
=N!\det\left[\int\varphi_i(x)\psi_j(x)d\mu(x)\right]_{i,j=1}^N\ ,
\end{split}
\end{equation}
to $\varphi_i(x)=x^{2(i-1)}$, $\psi_i(x)=e^{-\left(\lambda_i^{(2)}\right)^2x^{-2}}$, and
$d\mu(x)=e^{-\alpha(\mu)^2x^2}x^{2\nu-1}dx$ on ${\mathbb R}_+$, we obtain that integral (\ref{GXsquareDistribution4}) is 
equal
$$
{\mathcal I}=N!\det\left[\int\limits_0^{\infty}e^{-\alpha(\mu)^2x^2-\left(\lambda_j^{(2)}\right)^2x^{-2}}x^{2(i+\nu)-3}dx\right]_{i,j=1}^N.
$$
To compute the integral inside the determinant above we use the formula
\cite{GradshteinRyzhik} 8.432.6
$$
\int\limits_0^{\infty}x^{\nu-1}\exp\left[-x-\frac{\rho^2}{4x}\right]dx=2\left(\frac{\rho}{2}\right)^{\nu}K_{-\nu}(\rho),
$$
where $K_{-\nu}(\rho)=K_{+\nu}(\rho)$ is the modified Bessel function of the second kind. The result is that integral (\ref{GXsquareDistribution4}) is proportional to
$$
{\mathcal I}=\const\cdot\det\left[\left(\lambda_j^{(2)}\right)^{i+\nu-1}K_{i+\nu-1}\left(2\alpha(\mu)\lambda_j^{(2)}\right)\right]_{i,j=1}^N.
$$
We conclude that the joint density of the singular values of the matrix $X_1X_2$ with $\nu=M-N$
s given by
\begin{equation}\label{dd}
\begin{split}
P^{(N,M)}(X_1,X_2)dX_1dX_2=&\const\cdot\det\left[\left(\lambda_i^{(2)}\right)^{j-1}I_{j-1}\left(2\delta(\mu)\lambda_i^{(2)}\right)\right]_{i,j=1}^N\\
&\times
\det\left[\left(\lambda_i^{(2)}\right)^{j+\nu-1}K_{i+\nu-1}\left(2\alpha(\mu)\lambda_i^{(2)}\right)\right]_{i,j=1}^N\prod\limits_{j=1}^N\lambda_j^{(2)}d\lambda_j^{(2)}.
\end{split}
\end{equation}
Changing to squared singular values, $\left(\lambda_i^{(2)}\right)^2=y_i$, we obtain equation (\ref{TheMainJointProbabilityDensityFunction}) up to a normalisation constant $\const=1/Z_N$. 

In order to compute this constant we can apply again the Andr\'eief identity equation (\ref{AndId}), interpreting the left hand side as a 
probability measure, 
with the following choice of functions for the squared singular values $y_j$ of $X_1X_2$:
\begin{equation}\label{psij}
\psi_j(x)=x^{\frac{j}{2}}I_{j}\left(2\delta(\mu)\sqrt{x}\right),
\end{equation}
and
\begin{equation}\label{varphij}
\varphi_j(x)=x^{\frac{j+\nu}{2}}K_{j+\nu}\left(2\alpha(\mu)\sqrt{x}\right).
\end{equation}
After applying the integral identity
the requirement that this probability measure is normalised reads as follows, 
$$
1=Z_N^{-1}
N!\det\left[\int\limits_0^{\infty}y^{\frac{i+j+\nu}{2}}I_{i}\left(2\delta(\mu)\sqrt{y}\right)K_{j+\nu}\left(2\alpha(\mu)\sqrt{y}\right)dy\right]_{i,j=0}^{N-1}.
$$
The integral inside the determinant above can be computed explicitly. Namely, we have from \cite{GradshteinRyzhik} 6.576.7
\begin{equation}\label{IntegralIK}
\begin{split}
&\int\limits_0^{\infty}y^{\frac{i+j+\nu}{2}}I_{i}\left(2\delta(\mu)\sqrt{y}\right)K_{j+\nu}\left(2\alpha(\mu)\sqrt{y}\right)dy\\
&=\frac{1}{2}\alpha(\mu)^{j+\nu}\delta(\mu)^{i}\left(\alpha(\mu)^2-\delta(\mu)^2\right)^{-j-\nu-i-1}\Gamma(i+j+\nu+1).
\end{split}
\end{equation}
Taking into account the formula known from the normalisation of the Laguerre ensemble, 
$$
\det\left[\Gamma(i+j+\nu+1)\right]_{i,j=0}^{N-1}=\prod\limits_{j=1}^N\Gamma(j)\Gamma(j+\nu),
$$
we obtain the normalising constant (\ref{ZN}) in the formula for $P(y_1,\ldots,y_N)$.
The statement of Theorem \ref{PropositionTheMainJointProbabilityDensity} follows immediately.
\qed


\section{Proof of Theorem \ref{TheoremCorrelationKernelExact1}}
To derive an explicit formula for the correlation kernel of the biorthogonal ensemble under considerations we need the following Proposition.
\begin{prop}\label{CorrelationKernelBiorthogonal}
Let $\psi_j(x)$, $\varphi_j(x)$ be defined by equations (\ref{psij}) and (\ref{varphij}), where $j=0,1,\ldots,N-1$.
The correlation  kernel $K_N(x,y)$ of the biorthogonal ensemble defined by equation (\ref{detPP})
can be written as
\begin{equation}\label{KB}
K_N(x,y)=\sum\limits_{k,l=0}^{N-1}c_{k,l}\psi_k(x)\varphi_l(y),
\end{equation}
where the matrix $C=\left(c_{k,l}\right)_{k,l=0}^{N-1}$ is defined by
\begin{equation}\label{matrixG}
C=G^{-1},\;\; G=\left(g_{k,l}\right)_{k,l=0}^{N-1},\;\;g_{k,l}=\int\limits_0^{\infty}\psi_l(x)\varphi_k(x)dx.
\end{equation}
\end{prop}
\begin{proof} See Borodin \cite{BorodinBO}, Section 2.
\end{proof}
The matrix entries of $G$ can be computed explicitly. Using equation (\ref{IntegralIK}) the result is
\begin{equation}
g_{k,l}=\frac{1}{2}\alpha(\mu)^{k+\nu}\delta(\mu)^l\left(\alpha(\mu)^2-\delta(\mu)^2\right)^{-k-\nu-l-1}(k+l+\nu)!\ .
\end{equation}
This yields
\begin{equation}\label{ckl}
c_{k,l}=\frac{2\left(\alpha(\mu)^2-\delta(\mu)^2\right)^{k+l+\nu+1}}{\alpha(\mu)^{\nu+l}\delta(\mu)^{k}}a_{k,l},
\end{equation}
where  $\left(a_{k,l}\right)_{k,l=0}^{N-1}$ is the inverse of the Hankel matrix
\begin{equation}\label{Hankel}
H_{N-1}=\left(h_{k+l}\right)_{k,l=0}^{N-1},\;\;h_k=(k+\nu)!\ .
\end{equation}
Thus the problem of the computation of the correlation kernel is reduced to that of finding the inverse of the Hankel matrix $H_{N-1}$ defined by equation (\ref{Hankel}). 
A general method to find the inverse of a Hankel matrix can be described as follows.

Assume that there exists a probability measure $d\mu(x)$ on $\R$ such that all moments exist:
$$
h_k=\int x^kd\mu(x),\;\;k=0,1,\ldots .
$$
Construct the corresponding system $\{P_k\}$ of orthonormal polynomials,
$$
\int P_k(x)P_l(x)d\mu(x)=\delta_{k,l},\;\;k,l\geq 0.
$$
Consider the Christoffel-Darboux kernel $K_n(x,y)$,
$$
K_n(x,y)=\sum\limits_{k=0}^{n}P_k(x)P_k(y),
$$
rewrite this kernel in the form
$$
K_n(x,y)=\sum\limits_{i=0}^n\sum\limits_{j=0}^nq_{i,j}^{(n)}x^iy^j,
$$
and set
$$
Q_n=\left(q_{i,j}^{(n)}\right)_{i,j=0}^n.
$$
\begin{prop}\label{PropositionHankelInverse} We have
$$
H_nQ_n=I_n,
$$
where $H_n=\left(h_{i+j}\right)_{i,j=0}^{n}$, and $I_n$ is the unit matrix of order $n+1$.
\end{prop}
\begin{proof}
Using the reproducing property of the Christoffel-Darboux kernels we obtain
$$
\int x^kK_n(x,y)d\mu(x)=y^k,\;\;0\leq k\leq n.
$$
This can be rewritten as
$$
\int x^k\left(\sum\limits_{i=0}^n\sum\limits_{j=0}^nq_{i,j}^{(n)}x^iy^j\right)d\mu(x)
=\sum\limits_{i=0}^n\sum\limits_{j=0}^nq_{i,j}^{(n)}h_{k+i}y^j.
$$
Therefore,
$$
\sum\limits_{i=0}^n\sum\limits_{j=0}^nq_{i,j}^{(n)}h_{k+i}y^j=y^k,\;\;0\leq k\leq n.
$$
The equation just written above implies
$$
\sum\limits_{i=0}^nh_{k+i}q_{i,j}^{(n)}=\delta_{k,j},
$$
and the statement of the Proposition follows.
\end{proof}
\begin{prop} We have 
\begin{equation}\label{FormulaAkl}
a_{k,l}=\sum\limits_{p=0}^{N-1}\frac{(\nu+p)!(-p)_k(-p)_l}{p!(\nu+k)!k!(\nu+l)!l!}.
\end{equation}
\end{prop}
\begin{proof}Use Proposition \ref{PropositionHankelInverse}, and  observe that the relevant family of  orthogonal polynomials is that
of the classical Laguerre polynomials $\{L_n^{(\nu+k)}(x)\}$.  Then use the explicit formulae for $\{L_n^{(\nu+k)}(x)\}$ (see,  for example, \cite{GradshteinRyzhik} 8.970.1).
\end{proof}
After splitting factors accordingly among the functions $P_n(x)$ and $Q_n(x)$, 
including a factor of unity $(-1)^{n+n}$, 
formulae (\ref{KB}), (\ref{ckl}), and (\ref{FormulaAkl}) give us the expression for the correlation kernel stated in Theorem \ref{TheoremCorrelationKernelExact1}. Here we have also used that 
$(-n)_k=0$ for $k>n>0$.
\qed

\section{Proof of Theorem \ref{TheoremBorodinTypeFormula}}
In this Section we derive the formula for the correlation kernel $K_N(x,y)$ stated in Theorem \ref{TheoremBorodinTypeFormula} (equation \ref{KNMExplicitFormula1}). To obtain
equation (\ref{KNMExplicitFormula1}) from equations (\ref{K})-(\ref{Q}) of Theorem \ref{TheoremCorrelationKernelExact1} we use the following combinatorial fact.
\begin{prop}\label{prop6.1} Define $S(\alpha;k,r,N)$ by
\begin{equation}\label{Identity1}
S(\alpha;k,r,N)=\sum\limits_{n=0}^{N-1}\frac{n!}{(n-k)!(n-r)!}\Gamma(\alpha+n+1),
\end{equation}
where $N=1,2,\ldots;$ $k$, $r$ are two integers such that $0\leq k,r\leq N-1$, and $\alpha>-1$.
We have
\begin{equation}\label{Identity2}
S(\alpha;k,r,N)=\frac{(-1)^r\Gamma(\alpha+r+1)r!}{(N-1-k)!}\sum\limits_{i=0}^r\frac{\Gamma(N+i+\alpha+1)}{\Gamma(i+\alpha+1)}\frac{(-1)^i}{i!(r-i)!(\alpha+k+i+1)}.
\end{equation}
\end{prop}
\begin{proof} We will prove the equivalence of expressions (\ref{Identity1}) and (\ref{Identity2}) by induction with respect to $r$.
Namely, we will check that the equivalence   of expressions (\ref{Identity1}) and (\ref{Identity2}) takes place for $r=0$,
then we will assume that equation (\ref{Identity2}) is valid for an arbitrary $r$, and then we will show that this identity remains to be valid
when we replace $r$ by $r+1$.

From the definition (\ref{Identity1}) we have at $r=0$
$$
S(\alpha;k,r=0,N)=\sum\limits_{n=0}^{N-1}\frac{\Gamma(\alpha+n+1)}{\Gamma(n-k+1)}=\sum\limits_{n=0}^{N-k-1}\frac{\Gamma(\alpha+k+n+1)}{\Gamma(n+1)}.
$$
Using the formula
$$
\frac{\Gamma(a+s)}{\Gamma(s)}-\frac{\Gamma(a+s+1)}{\Gamma(s+1)}=-a\frac{\Gamma(a+s)}{\Gamma(1+s)},
$$
it is not hard to see that for $a>0$
$$
\sum\limits_{n=0}^L\frac{\Gamma(a+n)}{\Gamma(1+n)}=\frac{\Gamma(a+L+1)}{a\Gamma(L+1)}.
$$
Replacing $L$ by $N-k-1$, and $a$ by $\alpha+k+1>0$, we obtain
$$
S(\alpha;k,r=0,N)=\frac{\Gamma(N+\alpha+1)}{(\alpha+k+1)\Gamma(N-k)}.
$$
On the other hand, if $r=0$, then the right-hand side of equation (\ref{Identity2}) can be rewritten as
$$
\frac{\Gamma(\alpha+1)}{(N-1-k)!}\;\frac{\Gamma(N+\alpha+1)}{\Gamma(\alpha+1)(\alpha+k+1)}=\frac{\Gamma(N+\alpha+1)}{(\alpha+k+1)\Gamma(N-k)}.
$$
So the Proposition is proved for $r=0$. 

Using formula (\ref{Identity1}) we can obtain a recurrence relation for $S(\alpha;k,r,N)$, namely
\begin{equation}\label{recurrence}
S(\alpha;k,r+1,N)=S(\alpha+1;k,r,N)-(\alpha+r+1)S(\alpha;k,r,N).
\end{equation}
Now assume that formula (\ref{Identity2}) holds true for a certain $r\in{\mathbb N}$. In order to see that it remains to be valid for $r+1$ it is enough to show
that the right-hand side of equation  (\ref{Identity2}) satisfies equation (\ref{recurrence}). To see this, note  that the right-hand side of equation
(\ref{recurrence}) (with  $S(\alpha;k,r,N)$ given by equation (\ref{Identity2})) can be explicitly rewritten as
\begin{equation}\label{Difference}
\begin{split}
&\frac{(-1)^{r}\Gamma(\alpha+r+2)r!}{(N-1-k)!}\sum\limits_{i=0}^r\frac{\Gamma(N+i+\alpha+2)}{\Gamma(i+\alpha+2)}\frac{(-1)^i}{i!(r-i)!(\alpha+k+i+2)}\\
&-\frac{(-1)^r\Gamma(\alpha+r+2)r!}{(N-1-k)!}\sum\limits_{i=0}^r\frac{\Gamma(N+i+\alpha+1)}{\Gamma(i+\alpha+1)}\frac{(-1)^i}{i!(r-i)!(\alpha+k+i+1)}.
\end{split}
\end{equation}
Changing the index of summation in the first sum by one, $i\mapsto j=i+1$, we can rewrite expression (\ref{Difference}) as
\begin{equation}
\begin{split}
&\frac{(-1)^{r+1}\Gamma(\alpha+r+2)r!}{(N-1-k)!}\sum\limits_{j=1}^{r+1}\frac{\Gamma(N+j+\alpha+1)}{\Gamma(j+\alpha+1)}\frac{(-1)^jj}{j!(r+1-j)!(\alpha+k+j+1)}\\
&+\frac{(-1)^{r+1}\Gamma(\alpha+r+2)r!}{(N-1-k)!}\sum\limits_{j=0}^r\frac{\Gamma(N+j+\alpha+1)}{\Gamma(j+\alpha+1)}\frac{(-1)^j(r+1-j)}{j!(r+1-j)!(\alpha+k+j+1)}.
\end{split}
\nonumber
\end{equation}
Clearly, the sum of the two terms just written above can be represented as
\begin{equation}
\begin{split}
&\frac{(-1)^{r+1}\Gamma(\alpha+r+2)(r+1)!}{(N-1-k)!}\sum\limits_{j=0}^{r+1}\frac{\Gamma(N+j+\alpha+1)}{\Gamma(j+\alpha+1)}\frac{(-1)^j}{j!(r+1-j)!(\alpha+k+j+1)},
\end{split}
\nonumber
\end{equation}
which is $S(\alpha;k,r+1,N)$ as
given by equation (\ref{Identity2}). Thus we have seen that
the right-hand side of equation (\ref{Identity2}) satisfies equation (\ref{recurrence}). The Proposition is proved.
\end{proof}

Setting $\alpha=M-N$ and $r=l$ in Proposition \ref{prop6.1} and multiplying with $(-1)^{k+l}$ we obtain the following 
\begin{cor} The following identity holds true
\begin{equation}\label{CrucialSummationFormula}
\begin{split}
&\sum\limits_{p=0}^{N-1}\frac{(M-N+p)!}{p!}(-p)_k(-p)_l\\
&=\frac{(M-N+l)!l!}{(N-1-k)!}\sum\limits_{i=0}^l\frac{(i+M)!}{(M-N+i)!}\frac{(-1)^{i+k}}{i!(l-i)!(M-N+k+i+1)},
\end{split}
\end{equation}
where $M\geq N$.
\end{cor}
To get  equation (\ref{KNMExplicitFormula1}) for the correlation kernel $K_N(x,y)$ use formula
(\ref{CrucialSummationFormula}), and  equations (\ref{K})-(\ref{Q}) of Theorem \ref{TheoremCorrelationKernelExact1}.
\qed

\section{Proof of Proposition \ref{Biorthogonality}}
In this Section we begin to investigate the properties of the functions $P_n(x)$ and $Q_n(x)$ defined by
equations (\ref{P}) and (\ref{Q}). In particular,  we show that $P_n(x)$ and $Q_n(x)$  are biorthogonal functions.
To see this define two matrices, $V=\left(v_{k,p}\right)_{k,p=0}^{N-1}$ and $W=\left(w_{p,l}\right)_{p,l=0}^{N-1}$, by the formulae
\begin{equation}\label{vkp}
v_{k,p}=(-1)^p\frac{(\nu+p)!p!(-p)_k}{(\nu+k)!k!}\frac{\left(\alpha(\mu)^2-\delta(\mu)^2\right)^{k+\frac{1}{2}}}{\delta(\mu)^k},
\end{equation}
and
\begin{equation}\label{wpl}
w_{p,l}=(-1)^p\frac{2(-p)_l}{(p!)^2(\nu+l)!l!}\frac{\left(\alpha(\mu)^2-\delta(\mu)^2\right)^{l+\nu+\frac{1}{2}}}{\alpha(\mu)^{l+\nu}}.
\end{equation}
In addition, introduce two column vectors, ${\bf{\Psi}}(x)$, and ${\bf{\Phi}}(y)$,
\begin{equation}
{\bf{\Psi}}(x)=\left(\begin{array}{c}
                  \psi_0(x) \\
                  \psi_1(x) \\
                  \vdots \\
                  \psi_{N-1}(x)
                \end{array}
\right),\;\;\;
{\bf{\Phi}}(y)=\left(\begin{array}{c}
                  \varphi_0(y) \\
                  \varphi_1(y) \\
                  \vdots \\
                  \varphi_{N-1}(y)
                \end{array}
\right),
\nonumber
\end{equation}
where $\psi_j(x)$ and $\varphi_j(y)$ are defined by equations (\ref{psij}) and (\ref{varphij}).
Set
\begin{equation}\label{PQ}
{\bf{P}}(x)=V^{T}{\bf{\Psi}}(x),\;\; {\bf{Q}}(y)=W{\bf{\Phi}}(y).
\end{equation}
By elementary Linear Algebra calculations, the correlation kernel $K_{N}(x,y)$ equation (\ref{K})
can be written as
\begin{equation}\label{KernelBrief}
K_{N}(x,y)={\bf{P}}^T(x){\bf{Q}}(y).
\end{equation}
Observe that the matrix $G$ (defined by equation (\ref{matrixG})) can be written as
$$
G=\int\limits_{0}^{\infty}{\bf{\Phi}}(x){\bf{\Psi}}^T(x)dx.
$$
The notation above means that we integrate each matrix element of the $N\times N$ matrix
${\bf{\Phi}}(x){\bf{\Psi}}^T(x)$ from $0$ to $\infty$.

The matrix $C=\left(c_{k,l}\right)_{k,l=0}^{N-1}$
(whose matrix elements are given explicitly by equations (\ref{ckl}) and (\ref{FormulaAkl})) is the inverse of the matrix $G$. Therefore we can write
$$
C^{-1}=\int\limits_{0}^{\infty}{\bf{\Phi}}(x){\bf{\Psi}}^T(x)dx.
$$
The key observation is that
$$
C=VW,
$$
as it follows from  equations (\ref{ckl}), (\ref{FormulaAkl}), (\ref{vkp}), and (\ref{wpl}).
Since $C$ is invertible, both matrices $V$, $W$ are invertible, and we have
$$
\left(VW\right)^{-1}=\int\limits_{0}^{\infty}{\bf{\Phi}}(x){\bf{\Psi}}^T(x)dx,
$$
or
$$
I=\left(\int\limits_{0}^{\infty}{\bf{\Phi}}(x){\bf{\Psi}}^T(x)dx\right)VW.
$$
Multiplying both sides of the equation just written above by $W$ from the left, and using the definitions of the vectors
${\bf{P}}(x)$, ${\bf{Q}}(y)$ (see equation (\ref{PQ}) solved for ${\bf{\Phi}}(x)$ and ${\bf{\Psi}}^T(x)$) we obtain
$$
W=\left(\int\limits_{0}^{\infty}{\bf{Q}}(x){\bf{P}}^T(x)dx\right)W.
$$
Since the matrix $W$ is invertible, we conclude that
$$
\int\limits_{0}^{\infty}{\bf{Q}}(x){\bf{P}}^T(x)dx=I.
$$
In other words,  $P_n(x)$ and $Q_n(x)$  are biorthogonal functions. Proposition \ref{Biorthogonality} is proved.
\qed
\section{Proof of Proposition \ref{PropositionRecurrenceRelations} and Theorem \ref{TheoremChristoffelDarboux}}
In this Section we derive the recurrence relations for the functions $P_n(x)$ and $Q_n(y)$ stated in Proposition \ref{PropositionRecurrenceRelations}.
Using these recurrence relations we derive the Christoffel-Darboux type formula for the correlation kernel $K_N(x,y)$, and prove Theorem \ref{TheoremChristoffelDarboux}.
First, let us obtain equations (\ref{xPn1})-(\ref{xPn6}).
Setting
\begin{equation}
\hI_k(x)=\frac{k!x^{\frac{k}{2}}}{\delta(\mu)^k}I_k(2\delta(\mu)\sqrt{x}),\;\; k=0,1,\ldots\ ,
\end{equation}
the following recurrence relation holds true:
\begin{equation}\label{RRI}
x\hI_k(x)=\hI_{k+1}(x)+\frac{\delta(\mu)^2}{(k+1)(k+2)}\hI_{k+2}(x),\;\; \;\; k=0,1,\ldots
\end{equation}
To see this, use the recurrence relation for the Bessel functions, namely
$$
zI_{\nu}(z)=2(\nu+1)I_{\nu+1}(z)+zI_{\nu+2}(z).
$$
Introduce the vectors
$$
\mathbf{\hI}(x)=\left(\begin{array}{c}
          \hI_0(x)\\
          \hI_1(x)\\
          \vdots\\

       \end{array}
\right),\;\;\;
\mathbf{P}(x)=\left(\begin{array}{c}
               P_0(x) \\
               P_1(x) \\
               \vdots \\

             \end{array}
\right).
$$
The recurrence relations for the functions $\hI_k$ (equation(\ref{RRI})) can be rewritten as
\begin{equation}\label{xI}
x\mathbf{\hI}(x)=E\mathbf{\hI}(x),
\end{equation}
where the matrix $E$ is defined by the formula
\begin{equation}\label{MatrixE}
E_{k,m}=\delta_{k+1,m}+\frac{\delta(\mu)^2}{(k+1)(k+2)}\delta_{k+2,m};\;\;\; k,m=0,1,\ldots.
\end{equation}
Moreover, set
\begin{equation}\label{MatrixV}
\mathcal{V}_{p,k}=(-1)^p\frac{(\nu+p)!p!(-p)_k}{(\nu+k)!(k!)^2}
\left(\alpha(\mu)^2-\delta(\mu)^2\right)^{k+\frac{1}{2}},\;\;p,k=0,1,\ldots.
\end{equation}
Then we have
\begin{equation}\label{hP}
\mathbf{P}(x)=\mathcal{V}\mathbf{\hI}(x),
\end{equation}
where $\mathcal{V}=\left(\mathcal{V}_{p,k}\right)_{0\leq p,k\leq\infty}$.
From equations (\ref{xI}) and (\ref{hP}) we immediately obtain
\begin{equation}\label{xhP}
x\mathbf{P}(x)=\mathcal{V}E\mathbf{\hI}(x).
\end{equation}
Introduce the matrix $R_P$ by the formula
\begin{equation}\label{xhP1}
x\mathbf{P}(x)=R_P\mathbf{P}(x).
\end{equation}
The matrix $R_P$ is defining the recurrence relation for the functions $P_0(x)$, $P_1(x)$, $\ldots$ From equations (\ref{xI})-(\ref{xhP1})
we find
\begin{equation}\label{RPVEV}
R_P=\mathcal{V}E\mathcal{V}^{-1}.
\end{equation}
In the explicit calculations of the matrix $R_P$ below (and in the derivation of the recurrence relations)  we will exploit the following Lemma.
\begin{lem}For any 
non-negative
integers $i$,$j$ the following formulae hold true:
\begin{equation}\label{SUMMA1}
\sum\limits_{m=0}^{\infty}\frac{(-1)^{m+i}}{(i-m)!(m-j)!}=\delta_{i,j},
\end{equation}
\begin{equation}\label{SUMMA2}
\begin{split}
\sum\limits_{m=0}^{\infty}&\frac{(-1)^{m+i}(\nu+m+1)(m+1)^2}{(i-m)!(m+1-j)!}=(\nu+i+1)(i+1)^2\delta_{i+1,j}
\\
&
+\left(i^2+2i(\nu+i)+\nu+3i+1\right)\delta_{i,j}
+(\nu+3i)\delta_{i-1,j}+\delta_{i-2,j},
\end{split}
\end{equation}
\begin{equation}\label{SUMMA3}
\begin{split}
\sum\limits_{m=0}^{\infty}&\frac{(-1)^{m+i}(\nu+m+1)^2(m+1)}{(i-m)!(m+1-j)!}=
(i+1)(i+\nu+1)^2\delta_{i+1,j}\\
&+\left((\nu+i)^2+2(i+1)(i+\nu)+i+1\right)\delta_{i,j}+(2\nu+3i)\delta_{i-1,j}+\delta_{i-2,j},
\end{split}
\end{equation}
\begin{equation}\label{SUMMA4}
\begin{split}
\sum\limits_{m=0}^{\infty}&\frac{(-1)^{m+i}(m+1)(m+2)(\nu+m+1)(m+2)}{(i-m)!(m+2-j)!}=
(\nu+i+2)(\nu+i+1)(i+2)(i+1)\delta_{i+2,j}\\
&+2(i+1)(\nu+2i+2)(\nu+i+1)\delta_{i+1,j}+\left((\nu+i)(\nu+5i+3)+i(i+3)+2\right)\delta_{i,j}
\\
&+2(\nu+2i)\delta_{i-1,j}+\delta_{i-2,j}.
\end{split}
\end{equation}
\end{lem}
\begin{proof}Equation (\ref{SUMMA1}) is a reformulation of the fact that
$$
\left.(x-1)^{k-m}\right|_{x=1}=\delta_{k,m}.
$$
Equations (\ref{SUMMA2})-(\ref{SUMMA4}) can be derived using 
straightforward calculations. For example equation (\ref{SUMMA2}) 
is obtained by differentiating $\partial_x(x^{\nu+1}\partial_x(x\partial_x(x^j(x-1)^{i-j+1})))$ at $x=1$, after normalising by $(i-j+1)!$. Here $j=0,1$ have to be treated separately. The remaining equations follow in a similar fashion.
\end{proof}
\begin{prop} The matrix $\mathcal{V}$ is invertible, and its inverse is given by
\begin{equation}
\left(\mathcal{V}^{-1}\right)_{k,l}=\frac{(k!)^2(\nu+k)!}{(k-l)!(l!)^2(\nu+l)!}
\frac{1}{(\alpha(\mu)^2-\delta(\mu)^2)^{k+\frac{1}{2}}}
\;\;k,l=0,1,\ldots.
\end{equation}
\end{prop}
\begin{proof}This can be checked by direct calculations using formula (\ref{SUMMA1}).
\end{proof}
Equation (\ref{xhP1}) says that it is enough to compute the matrix $R_P$ explicitly to obtain the recurrence for $P_0(x)$, $P_1(x)$, $\ldots$ (equations (\ref{xPn1})-(\ref{xPn6})). This can be done exploiting formula (\ref{RPVEV}), the formula for the matrix elements of $\mathcal{V}$ (equation (\ref{MatrixV})), and that for the matrix elements of $E$ (equation (\ref{MatrixE})). In the computations we use formulae (\ref{SUMMA2}), (\ref{SUMMA4}) to express the sums involved in terms of the Kronecker symbols.

Now we turn to derivation of the recurrence relation for $Q_0(y)$, $Q_1(y)$, $\ldots$ (equations (\ref{yQn1})-(\ref{yQn6})).
Set
$$
\widehat{K}_l(y)=\frac{(l+\nu)!y^{\frac{l+\nu}{2}}}{\alpha(\mu)^{l+\nu}}K_{l+\nu}(2\alpha(\mu)\sqrt{y}).
$$
We have
\begin{equation}\label{RRK}
y\widehat{K}_l(y)=-\widehat{K}_{l+1}(y)+\frac{\alpha(\mu)^2}{(l+\nu+2)(l+\nu+1)}\widehat{K}_{l+2}(y).
\end{equation}
To see that equation (\ref{RRK}) holds true use the recurrence relations
$$
zK_{\nu}(z)=-2(\nu+1)K_{\nu+1}(z)+zK_{\nu+2}(z).
$$
Introduce the vectors
$$
\mathbf{\hK}(y)=\left(\begin{array}{c}
          \hK_0(y)\\
          \hK_1(y)\\
          \vdots\\

       \end{array}
\right),\;\;\;
\mathbf{Q}(y)=\left(\begin{array}{c}
               Q_0(y) \\
               Q_1(y) \\
               \vdots \\

             \end{array}
\right).
$$
Then the recurrence relation for the functions $\hK_k(y)$ (equation (\ref{RRK})) can be rewritten as
\begin{equation}\label{xK}
y\mathbf{\hK}(y)=\widetilde{E}\mathbf{\hK}(y),
\end{equation}
where the matrix $\widetilde{E}$ is defined by the formula
\begin{equation}\label{MatrixETILDA}
\widetilde{E}_{k,m}=-\delta_{k+1,m}+\frac{\alpha(\mu)^2}{(k+\nu+1)(k+\nu+2)}\delta_{k+2,m}\ ,\;\;\; k,m=0,1,\ldots.
\end{equation}
Moreover, set
$$
{\mathcal W}_{p,k}=\frac{2(-1)^p(-p)_k}{(p!)^2k!((\nu+k)!)^2}
\left(\alpha(\mu)^2-\delta(\mu)^2\right)^{\nu+k+\frac{1}{2}}
\ ,\;\;\; p,k=0,1,\ldots.
$$
We have
$$
\mathbf{Q}(y)={\mathcal W}\mathbf{\hK}(y).
$$
By the same argument as in the derivation of the recurrence relation for the functions $\hP_p(x)$ we find that the recurrence matrix $R_Q$ for the functions $Q_p(y)$ is given by
$$
R_{Q}={\mathcal W}\widetilde{E}{\mathcal W}^{-1}.
$$
\begin{prop}We have
$$
\left({\mathcal W}^{-1}\right)_{k,l}=\frac{((\nu+k)!)^2l!k!}{2(k-l)!}
\frac{1}{\left(\alpha(\mu)^2-\delta(\mu)^2\right)^{k+\nu+\frac{1}{2}}}.
$$
\end{prop}
\begin{proof}The formula for $\left({\mathcal W}^{-1}\right)_{k,l}$ can be obtained by direct calculations using formula (\ref{SUMMA1}).
\end{proof}
The subsequent computation leading to the recurrence relation for the functions $Q_0(y)$, $Q_1(y)$, $\ldots$
is very similar to that leading to the recurrence relation for the functions $P_0(x)$, $P_1(x)$, $\ldots$, where in the evaluation of the matrix $R_Q$ we use equations (\ref{SUMMA3}) and (\ref{SUMMA4}).
Proposition \ref{PropositionRecurrenceRelations} is proved.
\qed
\ \\

Now let us prove Theorem \ref{TheoremChristoffelDarboux}.  Setting $P_{-n}(x)=0=Q_{-n}(x)$ for $n=1,2$ we can apply the recurrence from Proposition \ref{PropositionRecurrenceRelations} as follows:
\begin{equation}
\begin{split}
(x-y)P_n(x)Q_n(y)=&\ \ a_{-2,n}P_{n-2}(x)Q_n(y)-a_{-2,n+2}P_{n}(x)Q_{n+2}(y)\\
&+a_{-1,n}P_{n-1}(x)Q_n(y)-a_{-1,n+1}P_{n}(x)Q_{n+1}(y)\\
&+a_{1,n}P_{n+1}(x)Q_n(y)-a_{1,n-1}P_{n}(x)Q_{n-1}(y)\\
&+a_{2,n}P_{n+2}(x)Q_n(y)-a_{2,n-2}P_{n}(x)Q_{n-2}(y)\ , \ \ \mbox{for}\ n=0,1,\ldots
\end{split}
\nonumber
\end{equation}
Here we have already used the relation between the coefficients  $a_{k,n}$ and $b_{k,n}$ in equation (\ref{RelationRecurrenceCoefficients}).
Summing up the right-hand side and the left-hand side of the equation above from $n=0$ to $N-1$ we obtain 
\begin{eqnarray}
(x-y)K_N(x,y)&=&\sum_{n=0}^{N-1}(a_{2,n}P_{n+2}(x)+a_{1,n} P_{n+1}(x))Q_n(y)+ \sum_{n=0}^{N-2}a_{-1,n+1} P_{n}(x)Q_{n+1}(y)\nonumber\\
&&+\sum_{n=0}^{N-3}a_{-2,n+2}P_{n}(x)Q_{n+2}(y)-\sum_{n=0}^{N-3}a_{2,n}P_{n+2}(x)Q_{n}(y)\nonumber\\
&&- \sum_{n=0}^{N-2}a_{1,n}P_{n+1}(x)Q_{n}(y)-\sum_{n=0}^{N-1}P_n(x)(a_{-1,n+1}Q_{n+1}(y)+a_{-2,n+2} Q_{n+2}(y)),
\end{eqnarray}
after shifting several summation indices. Cancelling all terms and dividing by $(x-y)\neq0$ 
 we obtain formula (\ref{CDKernel}) for the correlation kernel $K_{N}(x,y)$.
\qed

\section{Proof of Proposition \ref{PropositionContourIntegralRepresentationsPQ} and Theorem \ref{TheoremKContour}}
Let us first obtain the contour integral representation for the functions $P_0(x)$, $P_1(x)$, $\ldots$ as given in equation (\ref{ContourIntegralRepresentationsPn}).
Recall that $P_n(x)$ is given explicitly by equation (\ref{P}). We express the Bessel function in equation (\ref{P}) as an infinite sum,
$$
I_k\left(2\delta(\mu)x^{\frac{1}{2}}\right)=\sum\limits_{l=0}^{\infty}\frac{1}{l!(k+l)!}\left(\delta(\mu)x^{\frac{1}{2}}\right)^{k+2l}.
$$
Next we rewrite the formula for $P_n(x)$ as 
$$
P_n(x)=(-1)^n\frac{(\nu+n)!n!}{\nu!}\left(\alpha(\mu)^2-\delta(\mu)^2\right)^{\frac{1}{2}}
\sum\limits_{l=0}^{\infty}\frac{x^l\delta(\mu)^{2l}}{l!}\left(\sum\limits_{k=0}^n
\frac{(-n)_k\left(\alpha(\mu)^2-\delta(\mu)^2\right)^k}{(\nu+1)_kk!(k+l)!}x^k\right).
$$
The expression in the bracket on the right-hand side of the equation for $P_n(x)$ above can be written as a generalised 
hypergeometric series, so we have
\begin{equation}\label{PpContourIntegral1}
\begin{split}
P_n(x)=&(-1)^n\frac{(\nu+n)!n!}{\nu!}\left(\alpha(\mu)^2-\delta(\mu)^2\right)^{\frac{1}{2}}\\
&\times\sum\limits_{l=0}^{\infty}\frac{x^l\delta(\mu)^{2l}}{(l!)^2}{}_1F_2\left(\begin{array}{c}
               -n \\
               \nu+1, l+1
             \end{array}
             \biggl|\left(\alpha(\mu)^2-\delta(\mu)^2\right)x\right).
\end{split}
\end{equation}
The following contour integral representation can be obtained from residue calculus
$$
{}_1F_M\left(\begin{array}{c}
               -n \\
               1+\nu_1,\ldots,1+\nu_M
             \end{array}
             \biggl|x
\right)=\frac{(-1)^n\prod_{j=1}^M\Gamma(\nu_j+1)n!}{2\pi i}\oint\limits_{\Sigma}\frac{\Gamma(t+1)\Gamma(t-n)}{\prod_{j=1}^M\Gamma(t+\nu_j+1)}x^tdt,
$$
where $\Sigma$ is a closed contour that encircles $0,1,\ldots, n$ once in the positive direction.
In particular,
\begin{equation}
\begin{split}
&{}_1F_2\left(\begin{array}{c}
               -n \\
               \nu+1, l+1
             \end{array}
             \biggl|\left(\alpha(\mu)^2-\delta(\mu)^2\right)x\right)\\
&=\frac{(-1)^n\Gamma(\nu+1)\Gamma(l+1)n!}{2\pi i}
\oint\limits_{\Sigma}\frac{\Gamma(t-n)\left(\left(\alpha(\mu)^2-\delta(\mu)^2\right)x\right)^t}{\Gamma(t+1)\Gamma(t+\nu+1)\Gamma(t+l+1)}dt.
\end{split}
\nonumber
\end{equation}
Inserting the above formula into equation (\ref{P}), we obtain the desired expression for $P_n(x)$, equation (\ref{ContourIntegralRepresentationsPn}), after writing the remaining sum as another hypergeometric function.

Now we derive the contour integral representation for the functions $Q_0(y)$, $Q_1(y)$, $\ldots$ (equation (\ref{ContourIntegralRepresentationsQn})).
We start from the formula (\ref{Q}), and  use the relation (\ref{K=G}).
This enables us to rewrite equation  (\ref{Q}) as
\begin{equation}\label{QPG}
\begin{split}
Q_n(y)=&\frac{(-1)^n}{(n!)^2\nu!}\left(\alpha(\mu)^2-\delta(\mu)^2\right)^{\frac{1}{2}}\\
&\times\sum\limits_{l=0}^n\left(\frac{\alpha(\mu)^2-\delta(\mu)^2}{\alpha(\mu)^2}\right)^{l+\nu}
\frac{(-n)_l}{(\nu+1)_ll!}G^{2,0}_{0,2}\left(\begin{array}{cc}
                - \\
               0, l+\nu
             \end{array}\biggl|\alpha(\mu)^2y
\right).
\end{split}
\end{equation}
Following the definition (\ref{G-def}) a contour integral representation for the Meijer G-function in the formula above holds:
\begin{equation}\label{QPG1}
G^{2,0}_{0,2}\left(\begin{array}{cc}
                - \\
               0, l+\nu
             \end{array}\biggl|\alpha(\mu)^2y
\right)=\frac{1}{2\pi i}\int_{c-i\infty}^{c+i\infty}\Gamma(s+l+\nu)\Gamma(s)\left(\alpha(\mu)^2y\right)^{-s}ds,
\end{equation}
with $c>0$. Formulae (\ref{QPG}) and (\ref{QPG1}) result in the following expression for the function $Q_n(y)$
\begin{equation}
\begin{split}
Q_n(y)=&\frac{(-1)^n}{(n!)^2\nu!}\left(1-\frac{\delta(\mu)^2}{\alpha(\mu)^2}\right)^{\nu}\left(\alpha(\mu)^2-\delta(\mu)^2\right)^{\frac{1}{2}}\\
&\times\frac{1}{2\pi i}\int\limits_{c-i\infty}^{c+i\infty}\Gamma(s)\Gamma(s+\nu)
{}_2F_1\left(\begin{array}{c}
               -n, \nu+s\\
                1+\nu
             \end{array}
             \biggl|1-\frac{\delta(\mu)^2}{\alpha(\mu)^2}\right)(\alpha(\mu)^2y)^{-s}ds.
             \end{split}
\end{equation}
The (Gauss) hypergeometric function inside the integral above can be written as follows using \cite{GradshteinRyzhik} 9.131.2
for $n\in{\mathbb{N}}$
\begin{equation}
\begin{split}
{}_2F_1\left(\begin{array}{c}
               -n,  \nu+s\\
               1+\nu
             \end{array}
             \biggl|1-\frac{\delta(\mu)^2}{\alpha(\mu)^2}\right)
            =\frac{\Gamma(1+\nu)\Gamma(1-s+n)}{\Gamma(1+\nu+n)\Gamma(1-s)}\ 
             {}_2F_1\left(\begin{array}{c}
               -n,  \nu+s\\
               s-n
             \end{array}
             \biggl|\frac{\delta(\mu)^2}{\alpha(\mu)^2}\right).
\end{split}
\end{equation}
Applying this formula, and the fact that
\begin{equation}
\frac{\Gamma(1-s+n)}{\Gamma(1-s)}=(-1)^n\frac{\Gamma(s)}{\Gamma(s-n)},
\label{GammaIdentity}
\end{equation}
which can be shown using \cite{GradshteinRyzhik} 8.334.3, we obtain  equation (\ref{ContourIntegralRepresentationsQn}).  Proposition \ref{PropositionContourIntegralRepresentationsPQ} is proved.
\qed

To obtain the contour integral representation for the correlation kernel $K_N(x,y)$ given in Theorem \ref{TheoremKContour} we need the following Lemma.
\begin{lem} We have
\begin{equation}\label{CombinatorialLemma}
\begin{split}
\sum\limits_{n=0}^{N-1}\frac{\Gamma(t-n)}{\Gamma(s-n)}\frac{(-n)_k}{(s-n)_k}=&
k!\frac{\Gamma(t-N+1)}{\Gamma(s-t+k)}\sum\limits_{m=0}^k(-1)^m\left(\begin{array}{c}
                                                                                              N \\
                                                                                              m
                                                                                            \end{array}
\right)
\frac{\Gamma(s-t+m-1)}{\Gamma(s+m-N)}
\\
&-\frac{\Gamma(t+1)\Gamma(s-t-1)k!}{\Gamma(s)\Gamma(s-t+k)},
\end{split}
\end{equation}
where $k=0,1,\ldots, N-1$.
\end{lem}
\begin{proof} Denote by $S_N(t,s;k)$ the sum on the left hand side of equation (\ref{CombinatorialLemma}),
$$
S_N(t,s;k)=\sum\limits_{n=0}^{N-1}\frac{\Gamma(t-n)}{\Gamma(s-n)}\frac{(-n)_k}{(s-n)_k}.
$$
Also, set
$$
\widetilde{S}_N(t,s;k)=\sum\limits_{n=0}^{N-1}\frac{\Gamma(t-n)}{\Gamma(s-n)}\frac{n!}{(n-k)!}.
$$
These sums are related to each other according to the formula
\begin{equation}\label{SSTILDA}
    S_N(t,s;k)=(-1)^k\widetilde{S}_N(t,s+k;k).
\end{equation}
Thus it is enough to find a closed formula for $\widetilde{S}_N(t,s;k)$. Using the elementary property $x\Gamma(x)=\Gamma(x+1)$ is is easy to check that the following identity
holds true
\begin{equation}
\begin{split}
&\frac{\Gamma(t-n-1)}{\Gamma(s-n-1)}\frac{(n+1)!}{(n+1-k)!}-\frac{\Gamma(t-n)}{\Gamma(s-n)}\frac{n!}{(n-k)!}\\
&=k\frac{\Gamma(t-n-1)}{\Gamma(s-n-1)}\frac{n!}{(n-k+1)!}-(t-s)\frac{\Gamma(t-n-1)}{\Gamma(s-n)}\frac{n!}{(n-k)!}.
\end{split}
\end{equation}
This identity implies the following recurrence relation
$$
(s-t-1)\widetilde{S}_N(t,s;k)+k\widetilde{S}_N(t,s-1;k-1)=\frac{\Gamma(t-N+1)}{\Gamma(s-N)}\frac{N!}{(N-k)!},
$$
starting from $k=1,\ldots, N-1$. The recurrence relation above can be solved, and a formula for $\widetilde{S}_N(t,s;k)$
can be obtained. Namely, beginning with $k=0$,
$$
\widetilde{S}_N(t,s;k=0)=\sum\limits_{n=0}^{N-1}\frac{\Gamma(t-n)}{\Gamma(s-n)}=\frac{\Gamma(t-N+1)}{(s-t-1)\Gamma(s-N)}-\frac{\Gamma(t+1)}{(s-t-1)\Gamma(s)},
$$
which can be easily seen by induction in $N$ 
we find
\begin{equation}\label{STILDA}
\begin{split}
\widetilde{S}_N(t,s;k)=&\sum\limits_{l=0}^k\frac{\Gamma(t-N+1)(-k)_l}{(s-t-1)(s-t-2)\ldots (s-t-l-1)\Gamma(s-N-l)}\frac{N!}{(N-k+l)!}\\
&+\frac{(-1)^{k-1}\Gamma(t+1)k!}{(s-t-1)(s-t-2)\ldots (s-t-1-k)\Gamma(s-k)}.
\end{split}
\end{equation}
Formulae (\ref{SSTILDA}), (\ref{STILDA}) imply
\begin{equation}
\begin{split}
S_N(t,s;k)
=&\sum\limits_{l=0}^k\frac{(-1)^{k-l}\Gamma(t-N+1)\Gamma(s-t+k-l-1)k!}{\Gamma(s-t+k)(k-l)!\Gamma(s+k-l-N)}
\frac{N!}{(N-k+l)!}\\
&-\frac{\Gamma(t+1)k!\Gamma(s-t-1)}{\Gamma(s-t+k)\Gamma(s)}.
\end{split}
\nonumber
\end{equation}
Then, after setting $m=k-l$ we get formula (\ref{CombinatorialLemma}).
\end{proof}
Recall that the correlation kernel $K_{N}(x,y)$ can be represented as the sum of the biorthogonal functions $P_n(x)$ and $Q_n(y)$, see equation (\ref{K}).
We insert the integral representations for the functions $P_n(x)$ 
and 
$Q_n(y)$ (see Proposition \ref{PropositionContourIntegralRepresentationsPQ})  into equation (\ref{K}). 
We write the hypergeometric function ${}_2F_1$ as a finite sum up to $N-1$.
Then we interchange the finite sum and the double contour integral, use the combinatorial identity (\ref{CombinatorialLemma}),
and observe that the second term on the right-hand side of equation (\ref{CombinatorialLemma}) does not contribute to the  double contour integral.
The result of these calculations is the formula for the correlation kernel in the statement of  Theorem \ref{TheoremKContour}.
\qed
\section{Proof of Theorem \ref{TheoremHardEdgeScalingLimit}}
We use the contour integral representations for the functions $P_n(x)$ and $Q_n(y)$ obtained in  Proposition \ref{PropositionContourIntegralRepresentationsPQ}
together with the Christoffel-Darboux type formula for the correlation kernel $K_N(x,y)$, see Theorem \ref{TheoremChristoffelDarboux}.
Namely, we insert the contour integrals representing $P_n(x)$ and $Q_n(y)$ into formula (\ref{CDKernel}).
In the numerator of the right-hand side of equation (\ref{CDKernel}) we obtain a double contour integral.
Let us write this contour integral representation of the correlation kernel explicitly. We have for $N\geq2$
\begin{small}
\begin{equation}
\begin{split}
&K_{N}(x,y)=\frac{\alpha(\mu)^2-\delta(\mu)^2}{(2\pi i)^2(x-y)}\left(1-\frac{\delta(\mu)^2}{\alpha(\mu)^2}\right)^{\nu}\\
&\times\int\limits_{c-i\infty}^{c+i\infty}ds\oint\limits_{\Sigma}dt
\frac{\Gamma^2(s)\Gamma(s+\nu)(\alpha(\mu)^2-\delta(\mu)^2)^tx^t
{}_0F_1\left(\begin{array}{c}
                - \\
               t+1
             \end{array}\biggl|\delta(\mu)^2x
\right)}{\left(\Gamma(t+1)\right)^2\Gamma(t+\nu+1)}(\alpha(\mu)^2y)^{-s}\\
&\times\biggl\{-\ \frac{\alpha(\mu)^2}{\left(\alpha(\mu)^2-\delta(\mu)^2\right)^2}\frac{\Gamma(t-N+2)}{\Gamma(s-N)}
\;{}_2F_1\left(\begin{array}{c}
               -N,  \nu+s \\
               s-N
             \end{array}
             \biggl|\frac{\delta(\mu)^2}{\alpha(\mu)^2}\right)\\
&-\ \frac{\alpha(\mu)^2}{\left(\alpha(\mu)^2-\delta(\mu)^2\right)^2}\frac{\Gamma(t-N+1)}{\Gamma(s-N-1)}
\;{}_2F_1\left(\begin{array}{c}
               -N-1,  \nu+s \\
               s-N-1
             \end{array}
             \biggl|\frac{\delta(\mu)^2}{\alpha(\mu)^2}\right)\\
&-\left[\frac{3N+\nu}{\alpha(\mu)^2-\delta(\mu)^2}+\frac{2(2N+\nu)\delta(\mu)^2}{\left(\alpha(\mu)^2-\delta(\mu)^2\right)^2}\right]\frac{\Gamma(t-N+1)}{\Gamma(s-N)}
\;{}_2F_1\left(\begin{array}{c}
               -N,  \nu+s \\
               s-N
             \end{array}
             \biggl|\frac{\delta(\mu)^2}{\alpha(\mu)^2}\right)\\
&+\left[\frac{N^2(N+\nu)}{\alpha(\mu)^2-\delta(\mu)^2}+\frac{2N(N+\nu)(2N+\nu)\delta(\mu)^2}{\left(\alpha(\mu)^2-\delta(\mu)^2\right)^2}\right]\frac{\Gamma(t-N)}{\Gamma(s-N+1)}
\;{}_2F_1\left(\begin{array}{c}
               -N+1,  \nu+s \\
               s-N+1
             \end{array}
             \biggl|\frac{\delta(\mu)^2}{\alpha(\mu)^2}\right)\\
&+\frac{N(N-1)(N+\nu)(N+\nu-1)\delta(\mu)^2}{\left(\alpha(\mu)^2-\delta(\mu)^2\right)^2}
\frac{\Gamma(t-N)}{\Gamma(s-N+2)}
\;{}_2F_1\left(\begin{array}{c}
               -N+2,  \nu+s \\
               s-N+2
             \end{array}
             \biggl|\frac{\delta(\mu)^2}{\alpha(\mu)^2}\right)\\
&+\frac{N(N+1)(N+\nu)(N+\nu+1)\delta(\mu)^2}{\left(\alpha(\mu)^2-\delta(\mu)^2\right)^2}
\frac{\Gamma(t-N-1)}{\Gamma(s-N+1)}
\;{}_2F_1\left(\begin{array}{c}
               -N+1,  \nu+s \\
               s-N+1
             \end{array}
             \biggl|\frac{\delta(\mu)^2}{\alpha(\mu)^2}\right)\biggr\}.
\end{split}
\end{equation}
\end{small}
We note that as $N\rightarrow\infty$, we have the following
ratio asymptotic of Gamma functions
$$
\frac{\Gamma(t-N)}{\Gamma(s-N)}=\frac{\sin(\pi s)}{\sin(\pi t)}\frac{\Gamma(1-s+N)}{\Gamma(1-t+N)}
\simeq\frac{\sin(\pi s)}{\sin(\pi t)}N^{t-s}\left(1+O(N^{-1})\right),
$$
upon using equation (\ref{GammaIdentity}), $\Gamma(1-x)\Gamma(x)=\pi/\sin(\pi x)$, and the standard asymptotic expansion of the Gamma-function.
Moreover, we have \cite{Temme}
$$
{}_2F_1\left(\begin{array}{c}
               -N,  \nu+s \\
               s-N
             \end{array}
             \biggl|\frac{\delta(\mu)^2}{\alpha(\mu)^2}\right)
             \simeq\frac{1}{\left(1-\frac{\delta(\mu)^2}{\alpha(\mu)^2}\right)^{\nu+s}}\left(1+O(N^{-1})\right),
$$
as $N\rightarrow\infty$. 
Using the asymptotic formulae just written above, we find
\begin{equation}\label{KBeforeLimit}
\begin{split}
&\frac{1}{N\left(\alpha(\mu)^2-\delta(\mu)^2\right)}
K_{N}\left(\frac{x}{N\left(\alpha(\mu)^2-\delta(\mu)^2\right)},
\frac{y}{N\left(\alpha(\mu)^2-\delta(\mu)^2\right)}\right)\\
&=\frac{1}{(2\pi i)^2(x-y)}\int\limits_{c-i\infty}^{c+i\infty}ds\oint\limits_{\Sigma}dt\biggl[
\frac{\Gamma^2(s)\Gamma(s+\nu)}{\left(\Gamma(t+1)\right)^2\Gamma(t+\nu+1)}\frac{\sin\pi s}{\sin\pi t}\frac{x^t}{y^s}\\
&\times\left(A(s,t;N)+\frac{\delta(\mu)^2}{\alpha(\mu)^2-\delta(\mu)^2}B(s,t;N)\right)(1+O(N^{-1}))\biggr],
\end{split}
\end{equation}
where we used that the hypergeometric function $_0 F_1$ of rescaled argument tends to unity. 
The functions $A(s,t;N)$ and $B(s,t;N)$ are given by 
\begin{eqnarray}
A(s,t;N)&=&\frac{N^2(N+\nu)}{s-N}-(t-N)(s+t+N+\nu),\nonumber\\
B(s,t;N)&=&\frac{N(N+1)(N+\nu)(N+\nu+1)}{(t-N-1)(s-N)}+\frac{N(N-1)(N+\nu)(N+\nu-1)}{(s-N+1)(s-N)}\nonumber\\
&&+\frac{2N(N+\nu)(2N+\nu)}{s-N}-(t-N)(t+s+2N+2\nu).
\nonumber
\end{eqnarray}
Note that the additional factor in front of the kernel compensates the rescaling of the arguments of the factor $1/(x-y)$.
Computations show that
$$
\underset{N\rightarrow\infty}{\lim}A(s,t;N)=-s(s+\nu)-t(t+\nu)-st,
$$
and
$$
\underset{N\rightarrow\infty}{\lim}B(s,t;N)=0.
$$
Now we take the limit $N\rightarrow\infty$ from both sides of equation (\ref{KBeforeLimit}),
and interchange the limit and integrals in the right-hand side. The fact that we are allowed to take the limit inside the integrals
can be justified as in the proof of Theorem 5.3 in Kuijlaars and Zhang \cite{KuijlaarsZhang} using the dominated convergence theorem and the asymptotic properties of Gamma
functions.
Thus we obtain the limiting relation
\begin{equation}
\begin{split}
&\underset{N\rightarrow\infty}{\lim}\left\{\frac{1}{N\left(\alpha(\mu)^2-\delta(\mu)^2\right)}
K_{N}\left(\frac{x}{N\left(\alpha(\mu)^2-\delta(\mu)^2\right)},
\frac{y}{N\left(\alpha(\mu)^2-\delta(\mu)^2\right)}\right)\right\}\\
&=\frac{-1}{(2\pi i)^2(x-y)}\int\limits_{c-i\infty}^{c+i\infty}ds\oint\limits_{\Sigma}dt\biggl[
\frac{\Gamma^2(s)\Gamma(s+\nu)}{\left(\Gamma(t+1)\right)^2\Gamma(t+\nu+1)}\frac{\sin\pi s}{\sin\pi t}\frac{x^t}{y^s}
(s(s+\nu)+t(t+\nu)+st)\biggr].
\end{split}
\end{equation}
Since
$$
\frac{\Gamma(s)\sin\pi s}{\Gamma(t+1)\sin\pi t}=-\frac{\Gamma(-t)}{\Gamma(1-s)},
$$
we can rewrite the equation above as
\begin{equation}\label{ScalingLimit1}
\begin{split}
&\underset{N\rightarrow\infty}{\lim}\left\{\frac{1}{N\left(\alpha(\mu)^2-\delta(\mu)^2\right)}
K_{N}\left(\frac{x}{N\left(\alpha(\mu)^2-\delta(\mu)^2\right)},
\frac{y}{N\left(\alpha(\mu)^2-\delta(\mu)^2\right)}\right)\right\}\\
&=\frac{1}{(2\pi i)^2(x-y)}\int\limits_{c-i\infty}^{c+i\infty}ds\oint\limits_{\Sigma}dt\biggl[
\frac{\Gamma(-t)\Gamma(s)\Gamma(s+\nu)}{\Gamma(t+1)\Gamma(t+\nu+1)\Gamma(1-s)}\frac{x^t}{y^s}
(s(s+\nu)+t(t+\nu)+st)\biggr].
\end{split}
\end{equation}
It follows from the definition (\ref{G-def}) that
$$
\frac{1}{2\pi i}\oint\limits_{\Sigma}dt\frac{\Gamma(-t)}{\Gamma(t+1)\Gamma(t+\nu+1)}x^t=
-G^{1,0}_{0,3}\left(\begin{array}{ccc}
                       & - &  \\
                      0, & -\nu, & 0
                    \end{array}
\biggr|x\right),
$$
and that
$$
\frac{1}{2\pi i}\int\limits_{c-i\infty}^{c+i\infty}ds\frac{\Gamma(s)\Gamma(s+\nu)}{\Gamma(1-s)}y^{-s}
=G^{2,0}_{0,3}\left(\begin{array}{ccc}
                       & - &  \\
                      \nu, &0,& 0
                    \end{array}
\biggr|y\right).
$$
Now we can rewrite the right-hand side of equation (\ref{ScalingLimit1}) as
\begin{equation}\label{ScalingLimit2}
\begin{split}
\frac{f(x)\left(\nu y\frac{d}{dy}g(y)-\left(y\frac{d}{dy}\right)^2g(y)\right)}{x-y}
+\frac{x\frac{d}{dx}f(x)\left(-\nu g(y)+y\frac{d}{dy}g(y)\right)}{x-y}
-\frac{(x\frac{d}{dx})^2f(x)g(y)}{x-y},
\end{split}
\end{equation}
where
\begin{equation}\label{ScalingLimit3}
f(x)=G^{1,0}_{0,3}\left(\begin{array}{ccc}
                       & - &  \\
                      0, & -\nu, & 0
                    \end{array}
\biggr|x\right),\;\;
g(y)=G^{2,0}_{0,3}\left(\begin{array}{ccc}
                       & - &  \\
                      \nu, &0,& 0
                    \end{array}
\biggr|y\right).
\end{equation}
Expression (\ref{ScalingLimit2}) (with the functions $f(x)$, $g(y)$ defined by equation (\ref{ScalingLimit3}))
 gives the limiting kernel for the product of two matrices with independent complex Gaussian entries, see Proposition 5.4 in Kuijlaars and Zhang \cite{KuijlaarsZhang}.
 As it is shown in  Kuijlaars and Zhang \cite{KuijlaarsZhang} (see the proof of Theorem 5.3) such limiting kernel can be also
written as
$$
\int\limits_0^1G^{1,0}_{0,3}\left(\begin{array}{ccc}
                       & - &  \\
                      0, & -\nu, & 0
                    \end{array}
\biggr|ux\right)G^{2,0}_{0,3}\left(\begin{array}{ccc}
                       & - &  \\
                      \nu,&0, & 0
                    \end{array}
\biggr|uy\right)du.
$$
Theorem \ref{TheoremHardEdgeScalingLimit} is proved.
\qed
\section{Proof of Theorem \ref{CentralLimitTheorem}}\label{proofsFinite}
We use the following result for biorthogonal ensembles obtained by Breuer and Duits \cite{BreuerDuits}.
Assume we are given a biorthogonal ensemble on $\R_{\geq 0}$ defined by the joint probability density
function $P_N(x_1,\ldots,x_N)$. Assume further that the correlation kernel of this ensemble,  $K_N(x,y)$,
is given by
$$
K_N(x,y)=\sum\limits_{p=0}^{N-1}\psi_p^{(N)}(x)\phi_p^{(N)}(y),
$$
where the functions $\psi_p^{(N)}$, $\phi_k^{(N)}$ are orthonormal,
$$
\int\limits_{0}^{\infty}\psi_p^{(N)}(x)\phi_k^{(N)}(x)dx=\delta_{p,k}.
$$
Suppose we know that the functions $\psi_p^{(N)}$ satisfy a $2m+1$ term recurrence relation
\begin{equation}
\begin{split}
x\psi_n^{(N)}(x)=\sum\limits_{j=-m}^{m}a_{j,n}^{(N)}\psi_{n+j}^{(N)}(x),
\end{split}
\nonumber
\end{equation}
where $n=0,1,\ldots$, and $m$ is independent of $N$.
Here we define that $\psi_{-m}^{(N)}(x)=0,\ldots$, $\psi_{-1}^{(N)}(x)=0$.
In other words, there exists a banded matrix $J^{(N)}$ such that
$$
x\left(\begin{array}{c}
         \psi_{0}^{(N)}(x) \\
         \psi_{1}^{(N)}(x) \\
         \psi_{2}^{(N)}(x) \\
         \vdots
       \end{array}
\right)=J^{(N)}\left(\begin{array}{c}
         \psi_{0}^{(N)}(x) \\
         \psi_{1}^{(N)}(x) \\
         \psi_{2}^{(N)}(x) \\
         \vdots
       \end{array}\right).
$$
Let us consider the situation when the recurrence coefficients $a_{m,N}^{(N)}$, $a_{m-1,N}^{(N)}$,
$\ldots$, $a_{-m,N}^{(N)}$ have limits as $N\rightarrow\infty$, namely
$$
\underset{N\rightarrow\infty}{\lim}a_{m,N}^{(N)}=\alpha_m,\;\;
\underset{N\rightarrow\infty}{\lim}a_{m-1,N}^{(N)}=\alpha_{m-1},\ldots,
\underset{N\rightarrow\infty}{\lim}a_{-m,N}^{(N)}=\alpha_{-m}.
$$
In this situation we associate with $J^{(N)}$ a Laurent polynomial $s(w)$ defined by
$$
s(w)=\sum\limits_{j=-m}^m\alpha_jw^{j}.
$$
\begin{prop}\label{TheoremBreuerDuits}
Let $f$ be a polynomial with real coefficients, and define the linear statistics of the biorthogonal ensemble by the formula
$$
X_f^{(N)}=\sum\limits_{i=1}^Nf(x_i),
$$
where $x_1$, $\ldots$, $x_N$ are the points of the biorthogonal ensemble under considerations.
Then
$$
X_f^{(N)}-\mathbb{E}X_f^{(N)}\rightarrow {\mathcal N}\left(0,\sum\limits_{k=1}^{\infty}k\hat{f}_k\hat{f}_{-k}\right)
$$
in distribution, where
$$
\hat{f}_k=\frac{1}{2\pi i}\oint_{|w|=1}f(s(w))w^k\frac{dw}{w}.
$$
\end{prop}
\begin{proof} This statement is a corollary of a more general result for biorthogonal ensembles obtained by Breuer and Duits \cite{BreuerDuits}, see  
Theorem 2.1 and Corollary 2.2 therein.
\end{proof}
Note that since $f$ is a polynomial with real coefficients, $\hat{f}_k$ is real.

Now, let us consider the $N$-dependent probability distribution $\widetilde{P_{N,M}}(X_1,X_2)$ on the Cartesian
product of $\Mat(\C, N\times  M)$ and $\Mat(\C, M\times  N)$ defined by equation (\ref{PNEW}).
Let $y_1$, $\ldots$, $y_N$ be the squared singular values of the random matrix $X_1X_2$,
with its linear statistics given by
$$
Y_f^{(N)}=\sum\limits_{i=1}^Nf(y_i).
$$
By Theorem \ref{PropositionTheMainJointProbabilityDensity} the squared singular values $y_1$, $\ldots$, $y_N$ of the random matrix $X_1X_2$
form a biorthogonal ensemble on $\R_{\geq 0}$. The correlation kernel of this ensemble, $\widetilde{K_{N}}(x,y)$,
can be written as
$$
\widetilde{K_{N}}(x,y)=\sum\limits_{n=0}^{N-1}P_n'(x)Q_n'(y).
$$
The new functions, $P_n'(x)$ and $Q_n'(y)$, are defined in terms of $P_n(x)$ and $Q_n(y)$  as
$$
P_n'(x)=\frac{1}{n!(n+\nu)!}P_n(x),\;\;Q_n'(y)=n!(n+\nu)!Q_n(y),
$$
where $P_n(x)$ and $Q_n(y)$ are defined as previously  by equations (\ref{P}), (\ref{Q}), with $\alpha(\mu)$
replaced by $N\alpha(\mu)$, and $\delta(\mu)$ replaced by $N\delta(\mu)$. Clearly, the functions
$P'_p(x)$ and $Q'_m(x)$ are orthonormal,
$$
\int\limits_{0}^{\infty}P'_p(x)Q'_m(x)dx=\delta_{p,m}.
$$
Moreover, the 5 term recurrence relation ($m=2$ here) for the functions $P_n'(x)$ can be written as
\begin{equation}
x{P}_n'(x)=a_{2,n}'P_{n+2}'(x)+a_{1,n}'P_{n+1}'(x)+a_{0,n}'P_{n}'(x)
+a_{-1,n}'P_{n-1}'(x)+a_{-2,n}'P_{n-2}'(x).
\end{equation}
The coefficients $a_{2,n}'$, $a_{1,n}'$, $a_{0,n}'$, $a_{-1,n}'$, and $a_{-2,n}'$ easily follow from Proposition \ref{PropositionRecurrenceRelations} and 
are given explicitly by
\begin{eqnarray}
a_{2,n}'&=&\frac{\delta(\mu)^2}{\left(\alpha(\mu)^2-\delta(\mu)^2\right)^2}\frac{(n+\nu+1)(n+\nu+2)}{N^2},\\
a_{1,n}'&=&\frac{1}{\alpha(\mu)^2-\delta(\mu)^2}\frac{(n+1)(n+\nu+1)}{N^2}\nonumber\\
&&+
\frac{\delta(\mu)^2}{\left(\alpha(\mu)^2-\delta(\mu)^2\right)^2}\frac{2(2n+\nu+2)(n+\nu+1)}{N^2},\\
a_{0,n}'&=&\frac{1}{\alpha(\mu)^2-\delta(\mu)^2}\frac{3n^2+2\nu n+3n+\nu+1}{N^2}\nonumber\\
&&+\frac{\delta(\mu)^2}{\left(\alpha(\mu)^2-\delta(\mu)^2\right)^2}\frac{6n^2+6n\nu+\nu^2+6n+3\nu+2}{N^2},\\
a_{-1,n}'&=&\frac{1}{\alpha(\mu)^2-\delta(\mu)^2}\frac{n(3n+\nu)}{N^2}
+\frac{\delta(\mu)^2}{\left(\alpha(\mu)^2-\delta(\mu)^2\right)^2}\frac{2n(\nu+2n)}{N^2},\\
a_{-2,n}'&=&\frac{1}{\alpha(\mu)^2-\delta(\mu)^2}\frac{n(n-1)}{N^2}
+\frac{\delta(\mu)^2}{\left(\alpha(\mu)^2-\delta(\mu)^2\right)^2}\frac{n(n-1)}{N^2}.
\end{eqnarray}
Noting that
$$
\frac{1}{\alpha(\mu)^2-\delta(\mu)^2}=\mu,\;\;\frac{\delta(\mu)^2}{\left(\alpha(\mu)^2-\delta(\mu)^2\right)^2}=\frac{(1-\mu)^2}{4},
$$
we obtain
\begin{eqnarray}
&&\underset{N\rightarrow\infty}{\lim}a_{2,N}'=\alpha_2=\frac{(1-\mu)^2}{4},
\;\;
\underset{N\rightarrow\infty}{\lim}a_{1,N}'=\alpha_1=\mu+(1-\mu)^2,\nonumber\\
&&\underset{N\rightarrow\infty}{\lim}a_{0,N}'=\alpha_0=3\mu+\frac{3}{2}(1-\mu)^2,\;\;\underset{N\rightarrow\infty}{\lim}a_{-1,N}'=\alpha_{-1}=3\mu+(1-\mu)^2,
\nonumber\\
&&\underset{N\rightarrow\infty}{\lim}a_{-2,N}'=\alpha_{-2}=\mu+\frac{(1-\mu)^2}{4}.
\nonumber
\end{eqnarray}
Thus, Proposition \ref{TheoremBreuerDuits} can be applied, and the relevant Laurent polynomial can be computed explicitly.
The result follows.
\qed

\begin{appendix}

\section{Limits of the joint probability density function}\label{limitA}

In this appendix we derive the two limits $\mu\to1$ and $\mu\to0$ of the joint probability density function  $P(y_1,\ldots,y_N)$
equation (\ref{TheMainJointProbabilityDensityFunction}), as given in equations (\ref{FirstLimtingRelation}) and (\ref{SecondLimtingRelation}), respectively.

For the first limit  $\mu\to1$ leading to two independent Gaussian complex matrices we have 
$\delta(\mu)\rightarrow 0$, $\alpha(\mu)\rightarrow 1$. From the series representation of the function $I_{\kappa}(z)$, equation (\ref{BesselFunctionI}),
it is not hard to obtain the following limiting relation
$$
\underset{\mu\rightarrow 1}{\lim}\left(
\frac{\det\left[y_i^{\frac{j-1}{2}}I_{j-1}(2\delta(\mu)\sqrt{y_i})\right]_{i,j=1}^N}{\delta(\mu)^{\frac{N(N-1)}{2}}}\right)
=\frac{\det\left[y_i^{j-1}\right]_{i,j=1}^N}{\prod\limits_{j=1}^N\Gamma(j)}.
$$
The limit of the remainder of the pre-factor $Z_N$  and of the modified Bessel function 
of the second kind 
$K_{\kappa}(2\alpha(\mu)\sqrt{y})$ is trivial, and after expressing the latter in terms of the Meijer G-function from equation (\ref{K=G}) the limiting 
joint probability density function  $\lim_{\mu\to1}P(y_1,\ldots,y_N)$
in equation (\ref{FirstLimtingRelation}) follows.

In the second limit $\mu\rightarrow 0$ both $\delta(\mu)$ and $\alpha(\mu)$ diverge. Hence in  equation 
(\ref{TheMainJointProbabilityDensityFunction}) we have to replace the modified Bessel functions inside the determinants by their large argument asymptotic expressions.
Namely, we use the formulae
\begin{equation}\label{Bessasypm}
I_{\kappa}(z)\simeq \frac{e^z}{\sqrt{2\pi z}},\;\;\; K_{\kappa}(z)\simeq \sqrt{\frac{\pi}{2z}}e^{-z},
\end{equation}
see Gradshteyn and Ryzhik \cite{GradshteinRyzhik}, Section 8.45. This gives asymptotically 
\begin{equation}
\begin{split}
\det\left[y_i^{\frac{j+\nu-1}{2}}K_{j+\nu-1}\left(2\alpha(\mu)\sqrt{y_i}\right)\right]_{i,j=1}^N\simeq
\frac{\pi^{\frac{N}{2}}}{2^N\alpha(\mu)^{\frac{N}{2}}}
\prod\limits_{i=1}^Ny_i^{\frac{\nu}{2}-\frac{1}{4}}\exp\left[-2\alpha(\mu)y_i^{\frac{1}{2}}\right]
\det\left[y_i^{\frac{j-1}{2}}\right]_{i,j=1}^N,
\end{split}
\nonumber
\end{equation}
and
\begin{equation}
\begin{split}
\det\left[y_i^{\frac{j-1}{2}}I_{j-1}\left(2\delta(\mu)\sqrt{y_i}\right)\right]_{i,j=1}^N\simeq\frac{\pi^{\frac{N}{2}}}{2^N\delta(\mu)^{\frac{N}{2}}}
\prod\limits_{i=1}^Ny_i^{-\frac{1}{4}}\exp\left[2\delta(\mu)y_i^{\frac{1}{2}}\right]
\det\left[y_i^{\frac{j-1}{2}}\right]_{i,j=1}^N.
\end{split}
\nonumber
\end{equation}
Noting that
$$
\alpha(\mu)-\delta(\mu)=\frac{1+\mu}{2\mu}-\frac{1-\mu}{2\mu}=1,
$$
and that asymptotically 
$$
\alpha(\mu)\delta(\mu)=\frac{(1+\mu)}{2\mu}\frac{(1-\mu)}{2\mu}\simeq\frac{1}{4\mu^2},
$$
we obtain that the product of the two determinants in equation (\ref{TheMainJointProbabilityDensityFunction}) turns
into
$$
\frac{\mu^N}{2^N}
\left(\det\left[y_i^{\frac{j-1}{2}}\right]_{i,j=1}^N\right)^2
\prod\limits_{i=1}^Ny_i^{-\frac{1}{2}}
\exp\left[-2y_i^{\frac{1}{2}}\right].
$$
Moreover, as $\mu\rightarrow 0$, the normalising constant $Z_N$ in equation (\ref{TheMainJointProbabilityDensityFunction}) becomes asymptotically
equal to
$$
Z_N\simeq\frac{2^{NM}}{N!\mu^N\prod\limits_{j=1}^N\Gamma(j)\Gamma(j+\nu)}.
$$
Putting all these results together we obtain equation (\ref{SecondLimtingRelation}).

\end{appendix}



\begin{thebibliography}{99}

\bibitem{ABD} Akemann, G.; Baik, J.; DiFrancesco, P. (eds.) The Oxford Handbook of Random Matrix Theory, Oxford University Press, Oxford, 2011.

\bibitem{Akemann1}
Akemann, G.; Burda, Z. Universal microscopic correlation functions for products of independent Ginibre matrices. J. Phys. A: Math. Theor. 45 (2012)  465201.

\bibitem{AkemannIpsen} Akemann, G.; Ipsen, J.R.
Recent exact and asymptotic results for products of independent random matrices. arXiv:1502.01667 [math-ph].

\bibitem{AkemannIpsenKieburg}
Akemann, G.; Ipsen, J.; Kieburg M.
Products of rectangular random matrices: singular values and progressive scattering.  Phys. Rev. E 88 (2013) 052118.

\bibitem{AkemannIpsenStrahov} Akemann, G.; Ipsen, J.R.; Strahov, E.
Permanental processes from products of complex and quaternionic induced Ginibre ensembles.
Random Matrices: Th. Appl. 3 (2014), no. 4, 1450014.

\bibitem{AkemannKieburgWei}
Akemann, G.; Kieburg M.; Wei, L.
Singular value correlation functions for products of Wishart random matrices.
J. Phys. A. 46 (2013) 275205.

\bibitem{AkemannStrahov}
Akemann, G.; Strahov, E. Hole probabilities and overcrowding estimates for  products of complex Gaussian matrices.
J. Stat. Phys. 151(2013)  987-1003.

\bibitem{AlexeevGotzeTikhomirov}
 Alexeev, N.; G$\ddot{\mbox{o}}$tze, F.; Tikhomirov, A. Asymptotic distribution of singular values of powers of random matrices. Lith. Math. J. 50 (2010), no. 2, 121-132.
 
 \bibitem{Anderson}
Anderson, G.W.; Guionnet, A.; Zeitouni, O. An introduction to random matrices. Cambridge Studies in Advanced Mathematics, 118. Cambridge University Press, Cambridge, 2010.

\bibitem{Baha}
Balantekin, A.B., Character Expansions, Itzykson-Zuber Integrals, and the QCD Partition Function. Phys. Rev. D 62 (2000) 085017.

\bibitem{Bertola}
 Bertola, M.; Gekhtman, M.; Szmigielski, J. Cauchy-Laguerre two-matrix model and the Meijer-G random point field. Comm. Math. Phys. 326 (2014), no. 1, 111-144.

\bibitem{MB2} Bertola, M.; Bothner, T.
Universality conjecture and results for a model of several coupled positive-definite matrices. arXiv:1407.2597.

\bibitem{BorodinBO}
Borodin, A. Biorthogonal ensembles. Nuclear Phys. B  536  (1999),  no. 3, 704-732.

\bibitem{BreuerDuits}
Breuer, J.; Duits, M. Central Limit Theorems for Biorthogonal Ensembles and Asymptotics of Recurrence Coefficients.
arXiv:1309.6224v2.

\bibitem{Burda3}
Burda, Z.; Janik, R. A.; Waclaw, B. Spectrum of the product of independent random Gaussian matrices. Phys. Rev. E 81 (2010), no. 4, 041132.

\bibitem{CPV}  Crisanti, A.; Paladin, G.; Vulpiani, A. Products of Random Matrices. Springer, Heidelberg, 1993.

\bibitem{Deift1}
Deift, P.A. Integrable operators. In: V.~Buslaev, M.~Solomyak,
D.~Yafaev (eds.) Differential operators and spectral theory:
M.~Sh.~Birman's 70th anniversary collection. American Mathematical
Society Translations, ser. 2, 189, Providence, R.I., (1999).

\bibitem{Fischmann}
 Fischmann, J.; Bruzda, W.; Khoruzhenko, B.A.; Sommers, H.-J.; Zyczkowski, K.
 Induced Ginibre ensemble of random matrices and quantum operations. J. Phys. A 45 (2012), no. 7, 075203.
 
 \bibitem{ForresterLogGases}
 Forrester, P.J. Log-gases and random matrices. London Mathematical Society Monographs Series, 34. Princeton University Press, Princeton, NJ, 2010.

 \bibitem{ForresterProductWishart}
  Forrester, P.J. Eigenvalue statistics for product complex Wishart matrices. J. Phys. A 47 (2014), no. 34, 345202.

\bibitem{PeterMario} Forrester, P.J.; Kieburg, M.
  Relating the Bures measure to the Cauchy two-matrix model.
arXiv:1410.6883 [math-ph].

\bibitem{ForresterLiu} Forrester, P.J.; Liu, D.-Z.
Singular values for products of complex Ginibre matrices with a source: hard edge limit and phase transition.
arXiv:1503.07955 [math.PR].

\bibitem{ForresterWang} Forrester, P.J.; Wang, D.
Muttalib--Borodin ensembles in random matrix theory --- realisations and correlation functions. arXiv:1502.07147 [math-ph].

\bibitem{FK} Furstenberg, F.; Kesten, H. 
Products of random matrices.
Ann. Math. Stat. {\bf 31} (1960) 457-469.

 \bibitem{GotzeKostersTikhomirov}G$\ddot{\mbox{o}}$tze, F.; K$\ddot{\mbox{o}}$sters, H.; Tikhomirov, A.
 Asymptotic Spectra of Matrix-Valued Functions of Independent Random Matrices and Free Probability.  arXiv:1408.1732.
 
  \bibitem{GotzeNaumovTikhomirov}G$\ddot{\mbox{o}}$tze, F.; Naumov, A.; Tikhomirov, A.
 Distribution of Linear Statistics of Singular Values of the Product of Random Matrices.  arXiv:1412.3314.
 
\bibitem{GradshteinRyzhik}
Gradshteyn, I.S.; Ryzhik, I.M. Table of Integrals, Series, and Products.
A. Jeffrey and D. Zwillinger (eds.). Fifth edition (January 1994).

\bibitem{HC} Harish-Chandra. Differential operators on a semisimple Lie algebra, Am. J. Math. 79 (1957) 87-120.

\bibitem{IpsenKieburg} Ipsen, J.R.; Kieburg, M.
Weak Commutation Relations and Eigenvalue Statistics for Products of Rectangular Random Matrices.
Phys. Rev. E 89 (2014), 032106.

\bibitem{Its}
Its, A.R.; Isergin, A.G.; Korepin, V.E.;  Slavnov, N.A.
Differential equations for quantum correlation functions. Int. J.
Mod. Phys. {B 4} (1990) 1003--1037.

\bibitem{IZ} Itzykson C. and Zuber, J.B. The Planar Approximation. 2, J. Math. Phys. 21 (1980) 411.

\bibitem{TT1} Kanazawa, T.; Wettig, T.; Yamamoto, N.
Singular values of the Dirac operator in dense QCD-like theories.
JHEP 12 (2011) 007.

\bibitem{TT2} Kanazawa, T.; Wettig, T.
Stressed Cooper pairing in QCD at high isospin density: effective Lagrangian and random matrix theory.
JHEP 10 (2014) 055.


\bibitem{Kuijlaars} Kuijlaars, A.B.J.
Transformations of polynomial ensembles. arXiv:1501.05506 [math.PR]

\bibitem{KS} Kuijlaars, A.B.J.; Stivigny, D.
Singular values of products of random matrices and polynomial ensembles.	
Random Matrices: Th. Appl. {03} (2014) 1450011.

 \bibitem{KuijlaarsZhang}
Kuijlaars, A B.J.; Zhang, L.
Singular values of products of Ginibre random matrices, multiple orthogonal polynomials and hard edge scaling limits.
Comm. Math. Phys. 332 (2014) 759--781.

\bibitem{LeutwylerSmilga} Leutwyler H. and Smilga, A.
Spectrum of Dirac operator and role of winding number in QCD.
Phys. Rev. D 46 (1992) 5607.

\bibitem{Luke}
Luke, Y.L.  The special functions and their approximations. Academic Press, New York 1969.

\bibitem{Ralf} M\"uller, R.R. 
On the asymptotic eigenvalue distribution of concatenated vector-valued fading channels.
IEEE Trans. Inf. Theor. {Vol. 48} No. 7 (2002) 2086-2091

\bibitem{NIST} Olver, F.W.L et al. (eds.), NIST Handbook of Mathematical Functions. Cambridge University Press, Cambridge 2010.

\bibitem{Osborn} Osborn, J.C. Universal results from an alternate random matrix model for QCD with a baryon chemical potential.
Phys. Rev. Lett. 93 (2004) 222001.

\bibitem{Karol} Penson, K.A.; Zyczkowski, K.
Product of Ginibre matrices: Fuss-Catalan and Raney distributions. Phys. Rev. E {83} (2011) 061118.

\bibitem{Rourke}
 O'Rourke, S.; Soshnikov, A. Products of independent non-Hermitian random matrices. Electron. J. Probab. 16 (2011), no. 81, 2219–-2245.

\bibitem{Strahov}
 Strahov, E. Differential equations for singular values of products of Ginibre random matrices. J. Phys. A 47 (2014), no. 32, 325203.

\bibitem{Temme} 
 Temme, N. Large parameter cases of the Gauss hypergeometric function.
J. Comp. Appl. Math. 153 (2003) 441-462. 

 \bibitem{JT} Verbaarschot, J.J.M.; Wettig, T.
Random Matrix Theory and Chiral Symmetry in QCD.
Ann. Rev. Nucl. Part. Sci. {50} (2000) 343-410.

\end{thebibliography}
\end{document}